\documentclass[runningheads]{llncs}
\input{preamble.sty}
\begin{document}

\title{Quiescence in Concert:\\ Composing Multi-Channel Time-Outs for IOCO}

% \titlerunning{Multiple Times for Quiescence}

\ifanonymous
\author{%
  Anonymous Author(s)\inst{1,2}
}
\authorrunning{Anonymous et al.}
\institute{%
  Institution 1, Country\\
  \email{anonymous@example.com}
  \and
  Institution 2, Country\\
  \email{anonymous@example.com}
}
\else
\author{%
  Laura Brand\'an Briones\inst{2} \orcidlink{0009-0001-7402-0077} \and Petra van den Bos\inst{1} \orcidlink{0000-0002-9212-1525} \and Marcus Gerhold\inst{1}\orcidlink{0000-0002-2655-9617}
}

%\authorrunning{L. Brand\'an Briones et al.}

\institute{%
  Formal Methods and Tools, University of Twente, The Netherlands\\
  \email{\{m.gerhold, p.vandenbos\}@utwente.nl}
  \and
  FaMAF, Universidad Nacional de C\'ordoba, Argentina\\
  \email{laura.brandan@unc.edu.ar}
}
\fi
\maketitle

%% =====================================================================
\begin{abstract}
In Model-Based Testing (MBT), test suites are generated automatically from a formal specification. The theory of testing real-time systems is rich, but often underused in practice, partly because applying the timed machinery demands expertise practitioners should not need. In prior work we addressed this for timed testing with a canonic lifting operator, which lets a modeller specify behaviour as  plain labelled transition systems while implicitly obtaining the timed automata that express their quiescent behaviour (the explicit absence of outputs) with timers. 
This paper takes the next step: we show that our lifting still works when each component carries its own time-out on its own channel. This way, we introduce a multi-channel lifting. We show that it commutes with parallel composition, i.e. composing and lifting is the same as lifting first and then composing. We show that the MBT apparatus survives: conformance, test generation and verdicts are preserved by the lifting, on the \emph{testable} traces that a time-out based tester can observe.

\end{abstract}

%% =====================================================================
\section{Introduction}
\label{sec:intro}

Model-based testing (MBT) derives test cases automatically from a behavioural specification model and executes them against a system under test. The \emph{input--output conformance} relation \ioco~\cite{T08} is the de-facto standard correctness criterion for systems exhibiting nondeterministic behaviour. A distinctive feature of \ioco\ is its explicit treatment of \emph{quiescence}, i.e. the specified absence of outputs, commonly labelled $\delta$. 
In practice, testers detect quiescence through a time-out: if no output is observed within some finite time, the system is deemed quiescent. Several timed variants of \ioco\ reflect this, notably \tioco~\cite{LMN04}, \textbf{rtioco}~\cite{KT04} and the quiescence-focused $\tioco_M$~\cite{BBB04}.
In prior work~\cite{BvdBGS25-not-blind} we showed that these timed models need not be built by hand. Instead, the modeller specifies the system as a plain Labeled Transition System (LTS), and a canonic lifting~$\chi^M$ supplies the timing automatically by adding a clock and a global time-out~$M$. This produces a timed automaton and turns quiescence into a real time-out observation, as is commonly done in practice albeit without the formal underpinning. The lifting~$\chi^M$ comes with strong guarantees, as it preserves \ioco, commutes with test generation, and preserves test verdicts. In short, the modeller keeps working untimed and obtains the entire timed testing machinery for free.

This lifting, however, provides a \emph{single} global time-out, and real systems rarely admit only one.
To illustrate, suppose two components have been specified independently, each with its own natural time-out~$M_1$ and~$M_2$, and suppose $M_1 \ll M_2$. Running the composed system under the single-$M$ lifting of~\cite{BvdBGS25-not-blind} forces a choice. Taking $M = \min(M_1, M_2)$ declares the slower channel quiescent prematurely and so accepts implementations the specification rejects. Taking $M = \max(M_1, M_2)$ is sound but pointless, since every quiescence verdict on the faster channel is delayed until the slower channel's deadline. This inflates test-execution time needlessly. Even worse, any single intermediate $M$ inherits both downsides and collapses the two components' quiescence labels $\delta_1, \delta_2$ into a single~$\delta$, so the tester loses the ability to even tell which channel has gone silent. 

This motivates our current work: a multi-channel refinement of the lifting. It combines two existing pieces of work. The first is \emph{untimed} multi-channel \ioco~\cite{BHT98,H98}, where \emph{outputs} are partitioned into channels, each with its own quiescence, yielding the relation \mioco. The second is its \emph{timed} multi-channel counterpart $\mtioco{\Mdec}$ of Brand\'an Briones and Brinksma~\cite{BBB05}, which equips each channel with its own time-out.
Our paper provides the \emph{bridge} between them: a canonic lifting that lets the modeller continue to work on an untimed model and obtain the multi-channel timed theory for free. 
Moreover, this lifting is \emph{compositional}, meaning lifting the parallel composition of several components equals the parallel composition of the lifted components. A modeller can therefore build and time each component independently. 

\paragraph{Contributions.}
Concretely, our contributions are as follows.
\begin{itemize}
  \item[\(\circ \)] We introduce the multi-channel canonic lifting $\chidec$, which augments a labelled transition system with one clock and one quiescence time-out \emph{per output channel} (\Cref{sec:lifting}), generalising the single-time-out lifting of~\cite{BvdBGS25-not-blind};
  \item[\(\circ \)] We show that $\chidec$ bridges untimed multi-channel \ioco\ with the timed one $\ \mtioco{\Mdec}\ $ (\Cref{thm:preservation}):
    $\sys_I\:\ \tmioco\ \sys_S$ iff
    $\chidec(\sys_I)\:\ \mtioco{\Mdec}\ \chidec(\sys_S)$;
  \item[\(\circ \)] We show that $\chidec$ commutes with test-case generation and preserves test verdicts (\Cref{sec:tests});
  \item[\(\circ \)] We prove that $\chidec$ commutes with (shared-environment) parallel composition (\Cref{thm:compose}), i.e. $\chi^{\Mcat}(\sys_1\!\parop\!\sys_2)\!=\!\chi^{\Mvecone}(\sys_1)\!\parop\!\chi^{\Mvectwo}(\sys_2)$.
\end{itemize}

The extension is not a trivial generalisation of~\cite{BvdBGS25-not-blind}, since per-channel bounds may order quiescence observations in a suspension trace in a way no time-out based tester can produce in practice. To address the correspondence between the untimed and timed theories we define \emph{testable} traces (\Cref{def:TStraces}).
\paragraph{Paper overview.}
\Cref{sec:prelim} recalls labelled transition systems, and \Cref{sec:prelim:channels} fixes a multi-channel version of \ioco\ along the lines of~\cite{BHT98}.
\Cref{sec:timed} recalls timed automata and the timed multi-channel relation $\mtioco{\Mdec}$ along the lines of~\cite{BBB05}.
\Cref{sec:lifting} defines the lifting $\chidec$ and proves the conformance bridge.
\Cref{sec:tests} treats test generation and commutation.
\Cref{sec:compose} develops the compositionality theorem and the corresponding conformance corollary.
\Cref{sec:related,sec:conclusion} discuss related work and conclude.
This technical report includes proofs in \hyperlink{app:proofs-target}{Appendix~\ref*{app:proofs}}. 
\section{Labelled Transition Systems} 
\label{sec:prelim}

Labelled transition systems have transitions labelled with actions.
We fix a finite set of \emph{input actions} $\ActI$ and a finite set of \emph{output actions} $\ActO$ and write $\Act = \ActI \cup \ActO$. 
Inputs are suffixed with $?$, outputs with $!$, which are conventions on the naming of labels, not part of the labels themselves.
In LTSs, $\tau$ is often used to mark internal and invisible progress.
For the sake of simplicity, we exclude $\tau$-actions for now, though including it is not expected to cause any issues~\cite{BvdBGS26-not-blind}.

\begin{definition}[Labelled transition system]
\label{def:lts}
A \emph{labelled transition system} (LTS) is a tuple $\system = \ltstuple$, where $S$ is a finite set of states with unique initial state $s_0 \in S$ and $\to \subseteq S \times \Act \times S$ is the transition relation.
\begin{itemize}
    \item We write $\state \xrightarrow{\action} \state'$ for $(\state, \action, \state') \in \to$, and $\state\xrightarrow{\action}$ if $\state \xrightarrow{\action} \state'$ for some $\state'\in\states$ and $\state\not\!\!\!{\xrightarrow{\action}}$ if no such $\state'$ exists;
    \item For $\trace = \action_1 \cdots \action_n$ with $\action_i \in \Act$, we write $\state \xrightarrow{\trace} \state'$ if there exist states $\state_0, \state_1, \ldots, \state_n$ with $\state_0 = \state$, $\state_n = \state'$, and $\state_{i-1} \xrightarrow{\action_i} \state_i$
for all $1 \le i \le n$; we call $\trace \in \Act^*$ a \emph{trace}; 
    \item We write $\traces(\state)=\{\sigma\in\Act^*\mid \state\xrightarrow{\sigma}\}$  and set $\traces(\sys) = \traces(s_0)$.
\end{itemize}
\end{definition}

In \ioco\ theory the implementation model is assumed to be \emph{input-enabled}. The intuition is that a tester is always able to provide any input to the system at any given time. In our paper we make this distinction explicit by calling input-enabled systems \emph{input-output transition systems (IOTSs)}.

\begin{definition}[IOTS]
An \emph{input-output transition system} (IOTS) is an LTS $\system = \IOTStupla$ that is input-enabled, i.e.: 
$\forall\, s \in S: \forall\, i? \in \ActI: \state \transition{\inputlab?}.$
\end{definition}

\section{Multi-channel ioco}
\label{sec:prelim:channels}

Following~\cite{BHT98}, we partition the output alphabet into a fixed number of \emph{channels}. We will see later that each channel represents a component in a composed system.

\begin{definition}[Output channels]
\label{def:channels}
Let $n \geq 1$ for an $n\in\mathbb{N}$. An \emph{$n$-channel output partition} is a family $\{\ActO^{k}\}_{k=1}^{n}$ with $\ActO = \biguplus_{k=1}^{n} \ActO^{k}$, i.e. the channels cover $\ActO$ and are pairwise disjoint. We write $\Act^{k} = \ActI \cup \ActO^{k}$. The \emph{channel} of an output $o! \in \ActO$ is the unique $k$ with
$o! \in \ActO^{k}$.
\end{definition}

While Heerink~\cite{H98} partitions both inputs and outputs into channels, we partition \emph{only outputs} and keep inputs as a single shared set. This reflects the intuition that quiescence is an output observation: the tester controls when inputs are provided, but only observes when outputs are (resp.\ are not) emitted. Consequently, we later assume that inputs are shared among all components.

\begin{definition}[Per-channel quiescence]
\label{def:quiescence}
Let $\system$ be an LTS with an $n$-channel output partition. A state $\state \in \states$ is \emph{$\kchan$-quiescent} if there is no outgoing output transition in $\ActO^{\kchan}$ from $s$, i.e. $\forall\, o! \in \ActO^{\kchan}: s \ntransition{o!}$.
We introduce a fresh \emph{channel-indexed quiescence action}
$\delta_{k}$ for each channel $k \in \{1,\dots,n\}$ and write
\[
  \Act^\delta_O = \Act_O \cup \{\delta_1, \dots, \delta_n\},
  \qquad
  \Act^\delta   = \Act   \cup \{\delta_1, \dots, \delta_n\}.
\]
\end{definition}

Naturally, a state is quiescent in the classical sense~\cite{T08} iff it is $\kchan$-quiescent for every $\kchan$. We augment traces to include the $k$-quiescence labels $(\delta_1,\ldots,\delta_k)$ and call the resulting set suspension traces.
We follow our previous work~\cite{BvdBGS25-not-blind} closely to introduce more notation needed for multi-channel ioco.

\begin{definition}[Multi-channel ioco notation]
\label{def:mioco-notation}
Let $\system = \ltstuple$ be an LTS with an $n$-channel output partition, and $\state \in \states$. Then:
\begin{itemize}
  \item The \emph{channel-$\kchan$ outputs} of $\state$ are
  \begin{center}
    $\outop_{\kchan}(s) = \{o! \in \ActO^{\kchan} \mid s \transition{o!}\} \cup \{\delta_{\kchan} \mid s \text{ is $\kchan$-quiescent}\}$    
  \end{center}
  \item Given the empty sequence $\varepsilon$, $\action\in\Actd$ and $\sigma\in(\Actd)^*$ the \emph{states-after-trace} relation, extended from single actions $a$ to sequences $\sigma$, is:
    \begin{align*}
      s \afterop \varepsilon &= \{s\} \\
      s \afterop a           &=
        \{s' \in S\mid a \in \Act,\ s \transition{a} s'\}
        \cup \{s \mid a = \delta_{\kchan},\ s \text{ is $\kchan$-quiescent}\}\\
      s \afterop \action\;\sigma & = \bigcup\{\state'\afterop\sigma\mid\state'\in\state\afterop\action\}
    \end{align*} 
  \item The \emph{suspension traces} (i.e. traces explicitly including the $\delta_k$ labels) of $\state$ are
    $\Straces(\state) = \{\sigma \in (\Actd)^{\ast} \mid \state \afterop \; \sigma \neq \emptyset\}$,
    and $\Straces(\system) = \Straces(\state_0)$.
\end{itemize}
\end{definition}
This enables us to define multi-channel \ioco\, where the intuition is straightforward: rather than having one output channel there are multiple. Note that, unlike $\outop_k$ the operator $\afterop$ is channel-agnostic.
\begin{definition}[Multi-channel ioco]
\label{def:mioco}
Let $\sys_S$ be an LTS and $\sys_I$ an IOTS over the same set of inputs $\ActI$ and the same $n$-channel output 
partition. Then $\sys_I \;\mioco\; \sys_S$ iff
\[
  \forall\, \sigma \in \Straces(\sys_S): \forall\, \kchan \in \{1, \dots, n\}:
  \outop_{\kchan}(\sys_I \afterop \sigma)
  \subseteq
  \outop_{\kchan}(\sys_S \afterop \sigma).
\]
\end{definition}

\Cref{def:mioco} collapses to classical \ioco{} of~\cite{T08} when
$n = 1$, i.e. there is a single output channel and a single quiescence action $\delta = \delta_1$. The channel partition originates from~\cite{BHT98} and was later refined in~\cite{H98}; the per-channel quiescence labels $\delta_{\kchan}$ were later added in the timed setting of~\cite{BBB05}. Therefore, \Cref{def:mioco} can be considered the untimed restriction of the relation $\mtioco{\Mdec}$ of~\cite{BBB05}. 

\begin{example}[Display component]
\label{ex:display}
\Cref{fig:disp-lts} shows the UI component $\sys_{\textsc{UI}}$ of an ATM, a single-channel LTS with inputs $\ActI = \{\mathit{card?}, \mathit{pin?}\}$ and outputs $\ActO^{1} = \{\mathit{msg!}, \mathit{err!}\}$ holding status messages and error reports. It is an LTS but not an IOTS, since neither of the states accepts \emph{all} inputs. \Cref{fig:disp-lts}(b) adds the per-channel quiescence loops, i.e. $s_0$ and $s_2$ are 1-quiescent and have a $\delta_1$ loop, whereas $s_1$ and $s_3$ have an enabled output and no $\delta_1$-loop. In \Cref{sec:compose} we compose $\sys_{\textsc{UI}}$ with a cash dispenser to obtain a two-channel system.
\end{example}
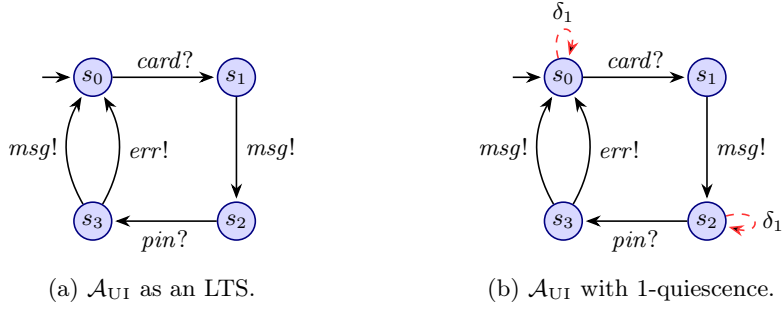
\begin{figure}[tbp]
  \centering
  \begin{subfigure}[b]{0.48\linewidth}
    \centering
    \begin{tikzpicture}[
      >=Stealth, shorten >=1pt, auto, node distance=1.9cm and 1.9cm,
      on grid, semithick,
      every state/.style={draw, circle, minimum size=5mm, inner sep=1pt,
        fill=blue!15, draw=blue!50!black},
    ]
      \node[state, initial, initial text=] (e0) {$s_0$};
      \node[state] (e1) [right=of e0] {$s_1$};
      \node[state] (e2) [below=of e1] {$s_2$};
      \node[state] (e3) [below=of e0] {$s_3$};
      \path[->]
        (e0) edge node {$\mathit{card?}$} (e1)
        (e1) edge node {$\mathit{msg!}$}  (e2)
        (e2) edge node {$\mathit{pin?}$}  (e3)
        (e3) edge [bend left=30] node {$\mathit{msg!}$} (e0)
        (e3) edge [bend right=30] node[swap] {$\mathit{err!}$} (e0);
    \end{tikzpicture}
    \caption{$\sys_{\textsc{UI}}$ as an LTS.}
    \label{fig:disp-lts-a}
  \end{subfigure}
  \hfill
  \begin{subfigure}[b]{0.48\linewidth}
    \centering
    \begin{tikzpicture}[
      >=Stealth, shorten >=1pt, auto, node distance=1.9cm and 1.9cm,
      on grid, semithick,
      every state/.style={draw, circle, minimum size=5mm, inner sep=1pt,
        fill=blue!15, draw=blue!50!black},
      delta/.style={dashed, draw=red!80},
    ]
      \node[state, initial, initial text=] (e0) {$s_0$};
      \node[state] (e1) [right=of e0] {$s_1$};
      \node[state] (e2) [below=of e1] {$s_2$};
      \node[state] (e3) [below=of e0] {$s_3$};
      \path[->]
        (e0) edge node {$\mathit{card?}$} (e1)
        (e1) edge node {$\mathit{msg!}$}  (e2)
        (e2) edge node {$\mathit{pin?}$}  (e3)
        (e3) edge [bend left=30] node {$\mathit{msg!}$} (e0)
        (e3) edge [bend right=30] node[swap] {$\mathit{err!}$} (e0);
      \path[->, delta]
        (e0) edge [loop above] node {$\delta_1$} (e0)
        (e2) edge [loop right] node {$\delta_1$} (e2);
    \end{tikzpicture}
    \caption{$\sys_{\textsc{UI}}$ with 1-quiescence.}
    \label{fig:disp-lts-b}
  \end{subfigure}
  \caption{The LTS $\sys_{\textsc{UI}}$ over the single channel $\ActO^{1} = \{\mathit{msg!}, \mathit{err!}\}$. Figure (a) is the plain LTS. Figure (b) adds the $\delta_1$ self-loops at the 1-quiescent states $s_0$ and $s_2$. States $s_1$ and $s_3$ are not 1-quiescent since \emph{msg!} (\emph{err!} resp.) is enabled.} 
  \label{fig:disp-lts}
\end{figure}
\section{Timed Automata and Multi-Channel \texorpdfstring{$\mtioco{\Mdec}$}{mtioco}}
\label{sec:timed}

We assume that the reader is familiar with timed automata following Alur~\cite{A99} and only briefly recall what we need here. Let $C$ be a finite set of \emph{clocks} and $\Phi(C)$ the set of \emph{clock constraints} generated by the grammar:
$$\varphi \Coloneqq c \sim K \mid \varphi_1 \wedge \varphi_2 \text{ for } c \in C,
K \in \mathbb{R}_{\geq 0}, \mathord{\sim} \in \{<, \leq, =, \geq, >\}.$$
As per usual, a \emph{clock valuation} $v : C \to \RRnonzero$ assigns each clock its current value. %\LB{here}
\begin{definition}[Timed automaton]
A \emph{timed automaton} (TA) is a tuple
$\systemTA = \langle L, \Act, \Phi_L, C, \to, \ell_0\rangle$ where $L$ is a finite set of locations with initial location $\ell_0 \in L$, $\Phi_L : L \to \Phi(C)$ assigns an \emph{invariant} to each location, and $\to \subseteq L \times \Act \times \Phi(C) \times 2^{C} \times L$ is the transition relation.
A transition $\langle \ell, a, \varphi, \lambda, \ell'\rangle$ is \emph{enabled} when its guard $\varphi$ holds; when a transition is taken the clocks in $\lambda$ reset to zero. We exclude Zeno behaviour, i.e. infinite transitions in a finite amount of time.
\begin{itemize}
    \item We write $\ell \transition{(d,\action)} \ell'$ for $(d,\action) \in \RRnonzero \times \Act$ if there is $\langle \ell, \action, \phi, \lambda, \ell' \rangle \in {\transition{}}$ such that $\phi$ and $\Phi(\location)$ are true for time $d$ that is spent between $\ell$ and $\ell'$, and such that $\Phi(\ell')$ is true after updating the clocks with the resets from $\lambda$;
    \item We lift $\transition{}$ to sequences, i.e. for $\ttrace = (d_1, \action_1) \cdots (d_n, \action_n) \in (\RRnonzero \times \Act)^*$, we write $\ell \transition{\ttrace} \ell'$ if there are locations $\ell_0, \dots, \ell_n$ with $\ell_0 = \ell$, $\ell_n = \ell'$, and $\ell_{i-1} \transition{(d_i, \action_i)} \ell_i$ for all $1 \le i \le n$. As before, $\ell \transition{\ttrace}$ means $\ell \transition{\ttrace} \ell'$ for some $\ell'$;
    \item Timed traces are sequences of non-negative numbers and visible actions, i.e. $\ttraces(\ell)=\{\ttrace \in(\RRnonzero\times\Act)^* \mid \ell\transition{\ttrace}\}$.
\end{itemize}
\end{definition}

Like before, we require an implementation to be input-enabled. %\LB{here}
For a TA, a location $\ell$ is \emph{input-enabled} if every input is enabled from $\ell$ at any $d$, i.e.\ $\forall\, i? \in \ActI: \ell \xrightarrow{(d, i?)}$. An \emph{input-output timed automaton} (IOTA) requires this while every channel is still below its quiescence bound. To formally quantify quiescence and anchor it in real-time we fix a vector $\Mdec = [M_1, \dots, M_n] \in \mathbb{R}^n_{>0}$. These $M_k$ later serve as per-channel quiescence bounds.  

\begin{definition}[IOTA]
\label{def:iota}
TA $\systemTA$ is an \emph{input-output timed automaton} (IOTA) for $\Mdec = [M_1,\dots,M_n] \in \RRplus^{n}$ if every location $\ell \in L$ is input-enabled under every time $d$: $\forall\, i? \in \ActI: \forall\, \ell \in L: \forall\, d \in \Rnonneg: d < M_k :\ \ell \xrightarrow{(d, i?)}$.

\end{definition}

The strict inequality $d < M_k$ for \emph{all} channels $k$ is deliberate, i.e. once a channel has reached its bound it has to conclude quiescence.  
Since this changes the suspension context, input-enabledness is required only before that point.

We now lift the abstract notion of per-channel quiescence from the untimed setting to TAs.
Unlike input-enabledness, \kchan-quiescence checks the \emph{absence} of \kchan-channel output transitions enabled at $\ell$. 

\begin{definition}[$\kchan$-quiescent]
\label{def:k-quiescent}
Let $\systemTA$ be a TA with output partition %\LB{here}\MG{Not $< M_k$?}
$\{\ActO^{\kchan}\}_{k=1}^{n}$ and bounds $\Mdec = [M_1, \dots, M_n] \in \mathbb{R}^n_{>0}$. A location $\ell \in L$ is \emph{$\kchan$-quiescent} if: 
$
 \forall\, d \in \RRnonzero: d < M_\kchan: \forall\ o! \in \ActO^{\kchan}:
    \ell \hspace{0.7em}\not\hspace{-0.7em}\xrightarrow{(d, o!)}.
$
\end{definition}
Unlike input-enabledness, $\kchan$-quiescence checks the \emph{absence} of $\kchan$-channel output transitions enabled at $\ell$, which no clock valuation can affect. We therefore state it via a delay $d$ rather than a valuation.

Below we introduce some notations that are needed to define $\mtioco{\Mdec}$.

\begin{definition}[$\mtioco{\Mdec}$ notation]
\label{def:mtioco-obs}
We define:
\begin{itemize}
\item For each channel $\kchan$:%\LB{here} 
\begin{align*}
  \outop^{\Mdec}_{\kchan}(\ell) = &\{(d, o!) \in (\RRnonzero\!\times\!\ActO^{k}) \mid \ell \xrightarrow{(d, o!)}\} \cup\\      
  &\{(d, \delta_{\kchan}) \mid \ell \text{ is $\kchan$-quiescent} \wedge d =M_{\kchan}\}
\end{align*}

\item As in \Cref{def:mioco-notation} we extend (and overload) $\afterop$ as follows:
     \begin{align*}
        \ell \afterop \epsilon = & \{ \ell \} \\
        \ell \afterop (\tim,\action) = &  \{\ell' \mid \action \in \Act \wedge\ \ell\transition{(\tim,\action)} \ell'\}\ \cup \\
         & \{\ell \mid (\tim,\action) = (M_{\kchan},\ \delta_{\kchan}) \wedge \ell \text{ is $\kchan$-quiescent}\} \\ 
         \ell \afterop (\tim,\action)\ttrace = & \bigcup\ \{\ell' \after \ttrace \mid  \ell'\in \ell \after (\tim,\action)\}
    \end{align*}

    \item We define the suspension timed traces as the traces of location $\ell$, including $\delta_k$ at time $M_k$, 
    for $k$-quiescent locations encountered in the timed trace:
    
    \[
        \SttracesM(\ell) = \{\ttrace \in (\RRnonzero \times \Actd)^* \mid \ell \after \ttrace \neq \emptyset \}.
    \]
        
    \item We write:  $\systemTA \after\ttrace=\ell_0\after\ttrace$ and $\SttracesM(\systemTA) = \SttracesM(\ell_0)$.
    
\end{itemize}
\end{definition}

We can now define the timed, multi-channel conformance relation $\mtioco{\Mdec}$.

\begin{definition}[$\mtioco{\Mdec}$]
\label{def:mioco-timed}
Let $\sysTAS$ be a TA and $\sysTAI$ an IOTA over the same output
partition $\{\ActO^{\kchan}\}_{k=1}^{n}$ with $\Mdec= [M_1, \dots, M_n]$. Then, $\sysTAI\ \mtioco{\Mdec}\ \sysTAS$ iff
$\forall\ \rho \in \SttracesM(\sysTAS) \forall\  \kchan\in\{1, \ldots, n\}:
   \outop^{\Mdec}_{\kchan}(\systemTAI \afterop \rho)
   \subseteq \outop^{\Mdec}_{\kchan}(\systemTAS \afterop \rho).$
\end{definition}

\Cref{def:mioco-timed} matches the relation introduced in~\cite{BBB05} \emph{in spirit}, up to cosmetic changes that align the notation with our previous paper~\cite{BvdBGS25-not-blind}. As in the untimed case, and contrary to~\cite{BBB05}, we assume a \emph{single shared input channel}.
\section{The Canonic Multi-Channel Lifting $\chidec$}
\label{sec:lifting}

We now define the lifting operator $\chidec$ in \autoref{def:lifting} as multi-channel counterpart of our previous work in~\cite{BvdBGS25-not-blind}, with one clock per channel, as first outlined in~\cite{BBB04}.
The full TA formalism (cf. \cite{A99}) admits intricate constructions that we do not need. The lifting produces only a restricted, \emph{canonic TA} that has \emph{exactly one clock per output channel}, location invariants of the shape $\bigwedge_{k} c_k \le M_k$, guards of the form $\bigwedge_{k} c_k < M_k$ (inputs), $c_k < M_k$ ($\kchan$-outputs) or $c_k = M_k$ ($\kchan$-quiescence), and resets of the form $\{c_k\}$ ($\kchan$-outputs or $\kchan$-quiescence) or $C$ (all clocks, for inputs).
An output on channel $\kchan$ resets \emph{only} $c_{\kchan}$, since it is no evidence of activity on any other channel $\kchan' \neq \kchan$. Inputs, by contrast, reset \emph{all} clocks. Since an input is global tester activity, once the tester provides a stimulus, quiescence is measured freshly on every channel and all their timers are reset.
The invariant forces each channel to conclude quiescence at exactly its own bound, i.e. once $c_{\kchan}$ reaches $M_{\kchan}$, no further time may elapse until $\delta_{\kchan}$ is observed. 
We represent $\chidec$ graphically in \autoref{tab:lifting}.
\autoref{cor:iots-to-iota} states that IOTS map to IOTA, as expected.

\begin{definition}[Multi-channel lifting $\chidec$]
\label{def:lifting}
Let $\sys = \langle S, \Act, \to, s_0\rangle$ be an LTS with $n$-channel output partition $\{\ActO^{\kchan}\}_{k=1}^{n}$, and $\Mdec = [M_1, \dots, M_n] \in
\mathbb{R}_{>0}^{n}$. 
The \emph{canonic multi-channel TA} of $\sysTA$ is the result of the mapping $\chidec:\mathit{LTS}\rightarrow\mathit{TA}$ such that
\[
  \chidec(\sys) = \langle L, \Actd, \Phi_L,
    C, \to_{\sys}, \ell_0\rangle  \text{ where }
\]

\begin{itemize}
\item $L = S$, with $\ell_0 = s_0$;
  \item $C = \{c_1, \dots, c_n\}$, one clock per output channel;
  \item $\Phi_L(\ell) = (\bigwedge_{\kchan=1}^{n} c_{\kchan} \leq M_{\kchan})$, the location invariants;
    \item  $\to_{\sys}: L \times \Act \times \Phi(C) \times 2^C \times L$ defines transition relation $\to_{\sys}$ as an extension of $\rightarrow$ with clock constraints and resets, as follows: 
    
    \begin{tabular}{r l l}
         $ \to_{\sys}\ =$ & $\{(\location,\action,\bigwedge_{\kchan=1}^n c_{\kchan}<M_{\kchan},C,\location') \mid (\location,\action,\location') \in\  \rightarrow,\  \action \in \ActI \}\ \cup$ \\
         & $\{(\location,\action, c_{\kchan}<M_{\kchan},\{c_k\},\location') \mid (\location,\action,\location') \in\  \rightarrow,\  \action \in \ActO \}\ \cup$\\
         & $\{(\location, \delta_k,c_k = M_k,\{c_k\},\location) \mid \location \in L \text{ is $k$-quiescent}\}.$
    \end{tabular}
    \end{itemize}
\end{definition}

%MGChanged; 'old' table below
\begin{table}[t]
\caption{Representation of $\chimult{\Mvec}$ (cf.~\Cref{def:lifting}) with $n = 2$ channels. Inputs reset all clocks; an output on channel $\kchan$ resets only $c_{\kchan}$; each $k$-quiescent location adds a quiescence self-loop $\delta_{\kchan}$ with guard $M_{\kchan}$, independently of the other channel.}
\label{tab:lifting}
\centering
\resizebox{\textwidth}{!}{ 
\renewcommand{\arraystretch}{1}
\setlength{\tabcolsep}{4pt}
\begin{tabular}{|c|c|c|c|c|}
\hline
  & \textbf{Input $i?$} 
  & \textbf{Output ch.~k $o_k!$} 
  & \textbf{Quiescence $\delta_{\kchan}$} \\ \hline

  \multirow{1}{*}{\rotatebox{0}{\scriptsize{\textbf{LTS}}}}
  &
  \begin{tikzpicture}[baseline=-1mm]
    \node[draw, circle, minimum size=4mm, inner sep=0pt] (a) at (0,0) {};
    \node[draw, circle, minimum size=4mm, inner sep=0pt] (b) at (1.4,0) {};
    \draw[->, >=Stealth] (a) -- node[above, font=\footnotesize] {$i?$} (b);
  \end{tikzpicture}
  &
  \begin{tikzpicture}[baseline=-1mm]
    \node[draw, circle, minimum size=4mm, inner sep=0pt] (a) at (0,0) {};
    \node[draw, circle, minimum size=4mm, inner sep=0pt] (b) at (1.4,0) {};
    \draw[->, >=Stealth] (a) -- node[above, font=\footnotesize] {$o_k!$} (b);
  \end{tikzpicture}
  &
  \begin{tikzpicture}[baseline=-1mm]
    \node[draw, circle, minimum size=4mm, inner sep=0pt] (a) at (0,0) {};
    \draw[->, >=Stealth] (a) edge[loop right, looseness=8] node[right, font=\footnotesize] {$\delta_{\kchan}$} (a);
  \end{tikzpicture}
  \\ \hline

  \multirow{1}{*}{\rotatebox{0}{\scriptsize{\textbf{TA after} $\chimult{\Mvec}$}}}
  &
  \begin{tikzpicture}[baseline=-1mm]
    \node[draw, rounded corners, font=\scriptsize, align=center, inner sep=2pt] (a) at (0,0) {$c_1 \le M_1$\\$\tiny{\wedge}$\\$c_2 \le M_2$};
    \node (b) at (3.2,0) {};
    \draw[->, >=Stealth] (a) -- node[above, font=\scriptsize] {\shortstack{$i?$\\$c_1\!<\!M_1 \wedge c_2\!<\!M_2$}} node[below, font=\scriptsize] {$\{c_1, c_2\}$} (b);
  \end{tikzpicture}
  &
  \begin{tikzpicture}[baseline=-1mm]
    \node[draw, rounded corners, font=\scriptsize, align=center, inner sep=2pt] (a) at (0,0) {$c_1 \le M_1$\\$\tiny{\wedge}$\\$c_2 \le M_2$};
    \node (b) at (2.6,0) {};
    \draw[->, >=Stealth] (a) -- node[above, font=\scriptsize] {\shortstack{$o_k!$\\$c_k\!<\!M_k$}} node[below, font=\scriptsize] {$\{c_k\}$} (b);
  \end{tikzpicture}
  &
    \begin{tikzpicture}[baseline=-1mm]
     \node[draw, rounded corners, font=\scriptsize, align=center, inner sep=2pt] (a) at (0,0) {$c_1\!\le\!M_1$\\ $\tiny{\wedge}$\\ $c_2\!\le\!M_2$};
    \draw[->, >=Stealth] (a) edge[loop right, looseness=4] node[right, font=\scriptsize] {\shortstack[l]{$\delta_{\kchan}$\\$c_{\kchan}=M_{\kchan}$\\$\{c_{\kchan}\}$}} (a);
  \end{tikzpicture}
  \\\hline
\end{tabular}}
\end{table}

\begin{restatable}[Input-enabledness under $\chidec$]{corollary}{corIotsIota}
\label{cor:iots-to-iota}
$\chidec$ maps IOTSs to IOTAs.
\end{restatable}

\begin{example}[Lifting the UI]
\label{ex:display-lifted}
We apply $\chidec$ to $\sys_{\textsc{ui}}$ with the single bound $M_1 = 1$, indicating that the user interface must conclude quiescence within one time unit. \Cref{fig:disp-ta} shows the result. Clock $c_1$ tracks UI quiescence and every location carries the invariant $c_1 \le 1$. Each output ($\mathit{msg!}, \mathit{err!}$) is guarded by $c_1 < 1$ and resets $c_1$, while the $\delta_1$ self-loops at $s_0$ and $s_2$ are enabled at exactly $c_1 = 1$ and reset $c_1$. Inputs reset $c_1$ as well indicating that the tester provided a stimulus. This is the single-channel construction of our prior work~\cite{BvdBGS25-not-blind}; the multi-channel structure is shown once we compose in~\Cref{sec:compose}.
\end{example}

%MGChanged, revert back to old below
\begin{figure}[t]
  \centering
  \begin{tikzpicture}[
    >=Stealth, shorten >=1pt, auto, node distance=2cm and 2cm,
    on grid, semithick,
    every state/.style={draw, align=center, rectangle, align=center, rounded corners=5pt, minimum height=4mm, minimum width=10mm, inner sep=1pt,
      fill=blue!15, draw=blue!50!black},
    el/.style={font=\scriptsize, align=center, inner sep=1pt},
    delta/.style={dashed, draw=red!80, font=\scriptsize, align=center},
  ]
    \node[state, initial, initial text=] (e0) {$s_0$\\ $\scriptstyle c_1\le1$};
    \node[state] (e1) [right=of e0, xshift=1cm] {$s_1$\\ $\scriptstyle c_1\le1$};
    \node[state] (e2) [below=of e1] {$s_2$\\ $\scriptstyle c_1\le1$};
    \node[state] (e3) [below=of e0] {$s_3$\\ $\scriptstyle c_1\le1$};
    \path[->]
      (e0) edge node[el] {$\mathit{card?}$\\$c_1\!<\!1,\ \{c_1\}$} (e1)
      (e1) edge node[el] {$\mathit{msg!}$\\$c_1\!<\!1,\ \{c_1\}$}  (e2)
      (e2) edge node[el] {$\mathit{pin?}$\\$c_1\!<\!1,\ \{c_1\}$}  (e3)
      (e3) edge [bend left=32] node[el] {$\mathit{msg!}$\\$c_1\!<\!1,\ \{c_1\}$} (e0)
      (e3) edge [bend right=32] node[el, swap] {$\mathit{err!}$\\$c_1\!<\!1,\ \{c_1\}$} (e0);
    \path[->, delta]
      (e0) edge [loop left] node {$\delta_1,c_1=1,\{c_1\}$} (e0)
      (e2) edge [loop right] node {$\delta_1,c_1=1,\{c_1\}$} (e2);
  \end{tikzpicture}
  \caption{The lifted UI component $\chi^{[1]}(\sys_{\textsc{UI}})$ with bound $M_1 = 1$. Every location carries the invariant $c_1 \le 1$. Outputs are guarded $c_1 < 1$ and reset $c_1$; the $\delta_1$ loops are enabled at $c_1 = 1$ at the UI-quiescent locations $s_0$ and $s_2$.}
  \label{fig:disp-ta}
\end{figure}
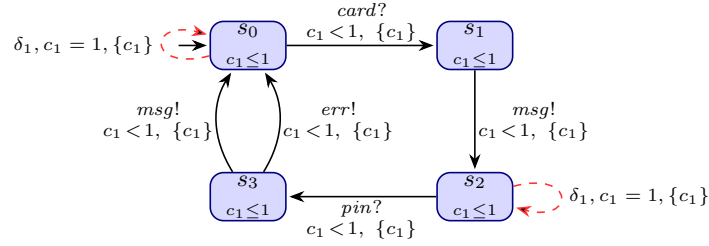

\paragraph{Preservation of Conformance.}
The core of our work is that the untimed \mioco\ relation (cf. \Cref{def:mioco}) matches the timed multi-channel relation (cf.~\Cref{def:mioco-timed}) after lifting. In other words, the lifting preserves conformance. To formally prove that, we need a language-theoretic result that both systems (before and after the lifting) share the same traces up to timing information. 

Intuitively, $\chidec$ makes outputs timed, but does not change which ones are enabled in a state, i.e. a channel-$\kchan$ output at a state becomes the timed observations $\{(d, o!) : d < M_{\kchan}\}$, 
%MGchange
and $\kchan$-quiescence becomes the single observation $(M_{\kchan},\delta_{\kchan})$. 
Untimed and timed observations thus correspond per channel. 
In particular, for a timed trace $\rho$, its \emph{projection} $\rho{\downarrow}$ removes all delays, i.e. for the empty sequence $\varepsilon$ it is $\varepsilon{\downarrow} = \varepsilon$ and otherwise inductively $((d, \action)\cdot\rho'){\downarrow} = \action \cdot (\rho'{\downarrow})$. Intuitively, $\rho{\downarrow} = \sigma$ means $\rho$ and $\sigma$ agree on actions and ignore timing (\autoref{lem:canonic}).

\begin{restatable}[Multi-channel canonic traces]{lemma}{lemCanonic}
\label{lem:canonic}
Let $\sys = \langle S, \Act, \to, s_0\rangle$ be an LTS with
an $n$-channel output partition and $\Mdec = [M_1, \dots, M_n] \in \RRplus^{n}$, then:\\
  If $\rho \in \SttracesM(\chidec(\sys))$, then there is
        $\sigma \in \Straces(\sys)$ such that
        $\rho{\downarrow}=\sigma$.
\end{restatable}

In the other direction, we need to be more careful. 
An LTS may have some traces, that are not testable after lifting to the timed setting. Consider a trace $\delta_2 o_1$ for $M_1 < M_2$: 
it models that output $o_1$ from channel 1, that should only be allowed before time-out $M_1$, could happen after time-out $M_2$ (due to the preceding $\delta_2$), which contradicts that output $o_1$ is allowed (before $M_1 < M_2$).
The lifted TA excludes this behaviour by making bounds explicit.
\autoref{def:TStraces} restricts suspension traces to only the testable traces.
Concretely, a suspension trace is testable iff its quiescence observations respect the ordering of the bounds and no $\delta_k$ occurs while a faster channel $k'$ ($M_{k'} < M_k$) has been silent for at least $M_{k'}$ without $\delta_{k'}$ being observed. The trace $\delta_2 o_1$ above violates this and describes an observer who has skipped a timeout that necessarily happened but was not observed by any timeout-based tester.

\begin{definition}[Testable traces] \label{def:TStraces}
Let $\sys$ be an LTS with
an $n$-channel output partition and $\Mdec = [M_1, \dots, M_n] \in \RRplus^{n}$. The set of \emph{testables traces} are:
 $$\TStraces(\sys) = \{\sigma \in \Straces(\sys) \mid \exists \rho \in \SttracesM(\chidec(\sys)): \rho{\downarrow} = \sigma\}.$$
\end{definition}

\Cref{def:TStraces} identifies the suspension traces that survive the lifting. Restricting \mioco\ to these traces yields a relation that, by \Cref{lem:canonic}, matches $\mtioco{\Mdec}$ exactly. 
\autoref{def:tmioco} introduces the restricted relation $\tmioco$, using outputs $\outT$, that only omits a $\delta_{\kchan}$ from $\outop$ that no realization of $\sigma$ can observe, i.e. when $\sigma$ is not a testable trace.
\begin{definition}[Testable multi-ioco]
\label{def:tmioco}
For an LTS $\sys$, $\sigma\in\Straces(\sys)$ and channel $\kchan$, let
$\outT_{\kchan}(\sys,\sigma) = \{a \in \ActO^{\kchan}\cup\{\delta_{\kchan}\} \mid \sigma\cdot a \in \TStraces(\sys)\}$.
Then $\sys_I \;\tmioco\; \sys_S$ iff
$\forall \sigma \in \TStraces(\sys_S)\ \forall \kchan:\ \outT_{\kchan}(\sys_I,\sigma) \subseteq \outT_{\kchan}(\sys_S,\sigma)$.
\end{definition}
\begin{restatable}[Preservation]{theorem}{thmPreservation}
\label{thm:preservation}
Let $\sys_I$ be an IOTS and $\sys_S$ an LTS over the same $n$-channel output partition. For every $\Mdec = [M_1, \dots, M_n] \in \mathbb{R}_{>0}^{n}$:
\begin{center}
  $\sys_I \;\tmioco\; \sys_S $ \ if and only if $\ \chidec(\sys_I) \;\mtioco{\Mdec}\; \chidec(\sys_S)$.
\end{center}
\end{restatable}
\section{Test Cases and Commutation}
\label{sec:tests}
We define multi-channel test cases for LTSs and TAs, respectively. Similar to our previous work~\cite{BvdBGS25-not-blind} test cases are inspired from the literature~\cite{T08,vdBS18}. 
We further connect these test cases to conformance by showing that verdicts are preserved by the lifting. An implementation passes every untimed test case iff its lifting passes every timed test case. 

As usual in \ioco-theory, the tester's choices at each state (resp. location) are: (1) stop; (2) observe outputs or quiescence; or (3) stimulate with an input. 
Our test cases here range over the augmented alphabet $\Actd = \Act \cup \{\delta_1, \dots, \delta_n\}$; the \emph{multi-channel} aspect is precisely that quiescence is observed per channel, via the labels $\delta_1, \dots, \delta_n$, rather than through a single $\delta$.
\begin{definition}[Multi-channel LTS test case]
\label{def:lts-test}
A \emph{test} for an LTS $\sys_S$ with $n$-channel output partition and bounds $\Mdec$ is a tree-shaped LTS $\test = \langle \states^\test, \Actd, \trans^\test,
\startingstate^\test \rangle$ satisfying:
\begin{itemize}
  \item $\test$ uses the same action labels and partitioning as $\sys_S$ plus
        $\delta_1, \dots, \delta_n$;
  \item $\test$ has only finite traces, is deterministic, and has
        no cycles;
  \item There are two special states $\pass, \fail \in \states^\test$;
  \item States $\pass$ and $\fail$ have no outgoing transitions;
  \item Every other state except $\pass$ and $\fail$ enables all outputs $\ActO$, and either one input \emph{or} all $\delta_{\kchan}$, i.e. $\forall\ \state \in \states^\test \setminus \{\pass, \fail\}:$\\
        $(|\inp(\state)|\!=\!0 \wedge
          \out(\state)\!=\!\ActO \cup \{\delta_1, \dots, \delta_n\})
         \vee
         (\out(\state)\!=\!\ActO \wedge
          |\inp(\state)|\!=\!1)$;
          
  \item \emph{Input-specifiedness:} All traces of $\test$ that end
        with an input are testable traces of $\sys_S$, i.e.\\
        $\forall\ \trace \in \traces(\test), \forall\ \iaction? \in \ActI:
         \trace \cdot \iaction? \in \traces(\test)
         \Trans
         \trace \cdot \iaction? \in\TStraces(\sys_S)$;
  \item \emph{Soundness:} All traces of $\test$ leading to $\pass$ are testable traces of $\sys_S$:\\
        $\forall\ \trace \in \traces(\test):
         \test \transition{\trace} \pass
         \Trans
         \trace \in \TStraces(\sys_S)$;
  \item \emph{Correctness:} All traces of $\test$ that end with an
        output or any $\delta_{\kchan}$, and lead to $\fail$,
        are \emph{not} testable traces of $\sys_S$:\\
        $\forall\ \trace \in\traces(\test),
         \forall\ \oaction \in \ActO \cup \{\delta_1, \dots, \delta_n\}:
         \trace \cdot \oaction \in \traces(\test)\ \wedge\
         \test \transition{\trace \cdot \oaction} \fail
         \Trans
         \trace \cdot \oaction \notin \TStraces(\sys_S).$
\end{itemize}
\end{definition}
Test cases thus depend on $\Mdec$ through $\TStraces$. Note that correctness also uses $\TStraces$, i.e. a trace that is not testable cannot be observed by a time-out based tester. A natural refinement in the multi-channel setting would allow \emph{``per-channel observe states''}, i.e. states enabling $\ActO \cup \{\delta_{\kchan}\}$ for a single channel $\kchan$. Operationally this lets a test focus on one channel (e.g. waiting out the cash channel's time-out without branching on display-channel observations it does not care about), which can yield smaller, more targeted tests.
The trade-off is that ignoring one channel shifts the burden of detecting the other channels' faults onto other tests. Hence, coverage must be recovered across the test suite rather than within each test.

As in ~\cite{BvdBGS25-not-blind} test verdicts are defined on traces alone; we do not parallel-compose tests with implementations. The parallel composition we use here (cf. \Cref{def:par-lts}) composes system specifications, not implementation models and tests as sometimes done in \ioco\ literature~\cite{TSB11}.
Here, verdicts are defined in a  lightweight way: $\sys_I$ fails $\test$ iff some trace $\sigma \in \TStraces(\sys_I) \cap \traces(\test)$ leads to $\fail$, and passes otherwise. We lift this to a set of tests (a test suite) $\mathcal{T}$ and say $\sys_I$ \emph{passes} $\mathcal{T}$ iff it passes every $\test \in \mathcal{T}$, and \emph{fails} $\mathcal{T}$ iff it fails some $\test \in \mathcal{T}$.

\begin{definition}[Multi-channel TA test case]
\label{def:ta-test}
A \emph{timed test case} for a canonic TA $\sysTA = \chidec(\sys_S)$ with an $n$-channel output partition is a tree-shaped TA $\testTA = \langle L^{t_{\mathit{TA}}}, \Actd, \Phi_L^{t_{\mathit{TA}}}, C, \trans^{t_{\mathit{TA}}}, \ell_0^{t_{\mathit{TA}}} \rangle$ satisfying:
\begin{itemize}
  \item $\testTA$ uses the same action labels and partitioning as $\sys_S$ plus
        $\delta_1, \dots, \delta_n$;
  \item $\testTA$ has timed traces using a finite number of actions, and has no cycles;
  \item every transition of $\testTA$ is \emph{reachable}, i.e.\ it occurs in some timed trace $\ttrace \in \ttraces(\testTA)$;
  \item There are two special locations
        $\pass, \fail \in L^{t_{\mathit{TA}}}$;
  \item Locations $\pass$ and $\fail$ have no outgoing transitions,
  \item $\testTA$ uses the clock set $C = \{c_1, \dots, c_n\}$, one
        per output channel. Every non-terminal location carries the
        canonic invariant
        $\bigwedge_{\kchan} c_{\kchan} \le M_{\kchan}$;
  \item Every non-terminal location enables all outputs $\ActO$ and
        either one input \emph{or} all per-channel quiescence labels
        $\delta_1, \dots, \delta_n$, refined per channel as follows:
        \begin{itemize}
          \item Each $\delta_{\kchan}$-transition carries guard
                $c_{\kchan} = M_{\kchan}$ and reset $\{c_{\kchan}\}$;
          \item Each input transition carries guard
                $\bigwedge_{\kchan} c_{\kchan} < M_{\kchan}$ and
                reset $C$;
          \item Each channel-$\kchan$ output transition carries guard
                $c_{\kchan} < M_{\kchan}$ and reset $\{c_{\kchan}\}$;
        \end{itemize}
  \item \emph{Input-specifiedness:} All timed traces of $\testTA$ that end with an input are suspension timed traces of
        $\sysTA$, i.e.\\
        $\forall\ \ttrace \in \ttraces(\testTA)\ \forall\ \tim \in \Rnonneg \
         \forall\ \iaction? \in \ActI:$\\
        $\ttrace \cdot (\tim, \iaction?) \in \ttraces(\testTA)
         \Trans
         \ttrace \cdot (\tim, \iaction?) \in \SttracesM(\sysTA);$
  \item \emph{Soundness:} All timed traces of $\testTA$ leading to
        $\pass$ are suspension timed traces of $\sysTA$:
        $\forall\ \ttrace \in \ttraces(\testTA):
         \testTA \transition{\ttrace} \pass
         \Trans \ttrace \in \SttracesM(\sysTA);$
  \item \emph{Correctness:} All timed traces of $\testTA$ that end
        with an output or $\delta_{\kchan}$, and lead to
        $\fail$ are \emph{not} suspension timed traces of
        $\sysTA$:\\
        $\forall\ \ttrace \in \ttraces(\testTA)\ \forall\ \tim \in \Rnonneg\ \forall\ \oaction \in \ActO \cup \{\delta_1, \dots, \delta_n\}:$\\
        $\ttrace \cdot (\tim, \oaction) \in \ttraces(\testTA)\ \wedge\
         \testTA \transition{\ttrace \cdot (\tim, \oaction)} \fail
         \Trans
         \ttrace \cdot (\tim, \oaction) \notin \SttracesM(\sysTA).$
\end{itemize}
\end{definition}

With test cases defined on both sides of the lifting, the two natural properties to investigate are (1) commutativity of the lifting operator (i.e. whether lifted test cases for the LTS are the same as test cases for the lifted LTS), and (2) the verdict preservation under the lifting.
Applied to a test case, $\chidec$ decorates the existing transitions, including the $\delta_{\kchan}$-transitions already present, and adds no self-loops, so $\pass$ and $\fail$ stay terminal. Strictly, this is another operator; for the sake of brevity we do not spell it out explicitly here.
\begin{restatable}[Test correspondence]{theorem}{thmTestCorr}
\label{thm:test-corr}
Let $\mathcal{T}_{\LTS}(\sys_S)$ be the set of test cases for an LTS $\sys_S$ (\Cref{def:lts-test}) and $\mathcal{T}_{\TA}(\chidec(\sys_S))$ the set of timed test cases for $\chidec(\sys_S)$ (\Cref{def:ta-test}). Then:
\[
  \chidec(\mathcal{T}_{\LTS}(\sys_S)) = \mathcal{T}_{\TA}(\chidec(\sys_S)).
\]
\end{restatable}

\Cref{thm:test-corr} shows that the lifting is commutative with respect to test generation, but it does not yet say anything about what running those tests \emph{means}.
The next theorem closes that gap. That is, passing or failing a test in the untimed paradigm agrees with passing or failing the lifted test in the timed paradigm.

\begin{restatable}[Verdict preservation]{theorem}{thmVerdicts}
\label{thm:verdicts}
For every IOTS $\sys_I$ and LTS $\sys_S$ and every $\Mdec = (M_1, \dots, M_n)\in\RRplus^n$:
\begin{enumerate}
  \item If $\sys_I$ passes $\mathcal{T}_{\LTS}(\sys_S)$,
    then $\chi^{\Mdec}(\sys_I)$ passes
    $\mathcal{T}_{\TA}(\chi^{\Mdec}(\sys_S))$.
  \item If $\sys_I$ fails $\mathcal{T}_{\LTS}(\sys_S)$,
    then $\chi^{\Mdec}(\sys_I)$ fails
    $\mathcal{T}_{\TA}(\chi^{\Mdec}(\sys_S))$.
\end{enumerate}
\end{restatable}
\Cref{thm:test-corr,thm:verdicts} prove the testing-side of the lifting, i.e. the modeller may continue to construct test cases for the untimed specification, and the verdicts agree with those of the lifted timed tests against the lifted implementation. 
\section{Compositionality}
\label{sec:compose}

The central result of this section is that the multi-channel lifting $\chidec$ commutes with parallel composition. Up to clock renaming, lifting a composed specification yields exactly the composition of the lifted specifications of components. This means a practitioner can model each component with its own per-channel time-outs and compose them for testing without losing conformance.

We compose components that operate in a \emph{shared environment}, i.e. each component reacts to the same (external) inputs, while producing outputs on its own disjoint channels. Our composition therefore  synchronises components on shared inputs and interleaves \emph{every} other action. Specifically, we assume that no output of one serves as an input to another as is commonly done in the paradigm of Lynch and Tuttle~\cite{LT88}. Their's is a richer setting that would interact subtly with quiescence here, since a blocked synchronised output can make a composed state quiescent where a single component is not. 

\begin{definition}[Shared-environment parallel composition of LTSs]
\label{def:par-lts}
Let $\sys_i=\langle S_i, \Act_i, \rightarrow_i, s_{0,i} \rangle$ be LTSs for $i = 1,2$ with disjoint output alphabets 
$\ActOi{1} \cap \ActOi{2} = \emptyset$ 
and shared inputs 
$\ActI^{\mathrm{sync}} = \ActIi{1} = \ActIi{2}$. 
Their \emph{parallel composition} $\sys_1 \parop \sys_2$ consists of the set of states $S_1 \times S_2$, initial state $(s_{0,1}, s_{0,2})$, and transitions
\begin{align*}
  (s_1, s_2) &\transition{a} (s_1', s_2')
    && \text{if } a \in \ActI^{\mathrm{sync}},\
       s_1 \transition{a} s_1',\ s_2 \transition{a} s_2'; \\
  (s_1, s_2) &\transition{a} (s_1', s_2)
    && \text{if } a \in \ActOi{1},\ s_1 \transition{a} s_1';\\
    (s_1, s_2) &\transition{a} (s_1, s_2') && \text{if } a \in \ActOi{2},\ s_2 \transition{a} s_2'.
\end{align*}

Its output partition is the disjoint union of the components' partitions, so each channel retains its own identity and bound.
\end{definition}

In the same vein we adapt parallel composition of TAs~\cite{A99}, which adds disjoint clock sets and component-wise joined invariants. Additionally, synchronisation on shared actions joins guards and unifies resets. Given $\Mdec_1=[M_1,\ldots,M_{n_1}] \in \Rpos^{n_1}$ and $\Mdec_2=[M'_1,\ldots,M'_{n_2}] \in \Rpos^{n_2}$, we write $$\Mcat=[M_1,\ldots,M_{n_1},M'_{1},\ldots,M'_{n_2}] \in \Rpos^{n_1+n_2}.$$ 

\begin{definition}[Shared-environment parallel composition of TAs]
\label{def:par-ta} 
Let \sloppy{$\sysTA_i = \langle L_i, \Act_i, \Phi_{L_i}, C_i, \to_i, \ell_{0,i}\rangle$} be canonic TAs for $i=1,2$ with disjoint clock sets, and output sets $\ActOi{1} \cap \ActOi{2} = \emptyset$, and shared inputs $\ActI^{\mathrm{sync}}=\ActIi{1}=\ActIi{2}$. 
Then $\sysTA_1 \parop \sysTA_2$ has location set $L_1 \times L_2$, clock set $C_1 \cup C_2$, initial location $(\ell_{0,1}, \ell_{0,2})$, 
invariants $\Phi_{L_1}(\ell_1) \wedge \Phi_{L_2}(\ell_2)$ for $(\ell_1,\ell_2) \in L_1 \times L_2$, 
and transitions:
\begin{itemize}[labelwidth=0.2em, labelsep=0.3em, leftmargin=*]
  \item synchronised, for $a \in \ActI^{\mathrm{sync}}$:
        $\langle (\ell_1,\ell_2), a, \varphi_1\wedge\varphi_2,
        \lambda_1\cup\lambda_2, (\ell_1',\ell_2')\rangle$ whenever
        $\langle\ell_i, a, \varphi_i, \lambda_i, \ell_i'\rangle\in\,\to_i$;
%    \item interleaved, for $a \in \ActOi{1} \cup \{\delta_{\kchan} \mid \kchan \in \sysTA_1\}$:\\
%        $\langle (\ell_1,\ell_2), a, \varphi_1, \lambda_1, (\ell_1',\ell_2)\rangle$
%        whenever $\langle\ell_1, a, \varphi_1, \lambda_1, \ell_1'\rangle\in\,\to_1$;
%  \item interleaved, for $a \in \ActOi{2} \cup \{\delta_{\kchan} \mid \kchan \in \sysTA_2\}$:\\
%        $\langle (\ell_1,\ell_2), a, \varphi_2, \lambda_2,
%        (\ell_1,\ell_2')\rangle$ whenever
%        $\langle\ell_2, a, \varphi_2, \lambda_2, \ell_2'\rangle\in\,\to_2$.
    \item interleaved,$\forall a\!\in\! \ActOi{1}^{\delta_1}$:\!
        $\langle (\ell_1,\ell_2), a, \varphi_1, \lambda_1, (\ell_1',\ell_2)\rangle$ when $\langle\ell_1, a, \varphi_1, \lambda_1, \ell_1'\rangle\in\,\to_1$;
  \item interleaved,$\forall a\!\in\ActOi{2}^{\delta_2}$:$\langle (\ell_1,\ell_2), a, \varphi_2, \lambda_2,(\ell_1,\ell_2')\rangle$ when $\langle\ell_2, a, \varphi_2, \lambda_2, \ell_2'\rangle\!\in\,\to_2$.
\end{itemize}
\end{definition}

\begin{example}[Composing dispenser and display]
\label{ex:atm-compose}
The full ATM arises by composing $\sys_{\textsc{ui}}$ with a cash dispenser $\sys_{\textsc{disp}}$, a single-channel LTS over $\ActO^{2} = \{\mathit{money!}\}$ that accepts $\mathit{card?}$ then $\mathit{pin?}$ before dispensing, see~\Cref{fig:compose}(a). \Cref{fig:compose}(b) and (c) show the composition before and after the lifting. The components synchronise on the shared inputs $\mathit{card?}, \mathit{pin?}$ and have disjoint output channels with bounds $\Mcat=[1, 5]$.
In other words, the UI component concludes quiescence in one time unit whereas the dispenser in takes five. 
The composed state $(s_2, e_3)$ has \emph{no} quiescence loop, as both channels have an output enabled, whereas $(s_0, e_0)$ and $(s_1, e_2)$ are quiescent on both channels and carry $\delta_1$ and $\delta_2$ enabled \emph{at the different times}. 
\end{example}

\begin{figure}[t]
  \vspace{-0.0cm}
  \centering
  % ===== TOP ROW: (a) dispenser, full width, centred =====
  \begin{subfigure}[b]{\linewidth}
    \centering
    \begin{tikzpicture}[
      >=Stealth, shorten >=1pt, auto, node distance=3cm,
      on grid, semithick,
      every state/.style={draw, circle, minimum size=5mm, inner sep=1pt,
        fill=blue!15, draw=blue!50!black},
      delta/.style={dashed, draw=red!80},
    ]
      % drawn horizontally to keep the row short
      \node[state, initial, initial text=] (d0) {$e_0$};
      \node[state] (d1) [right=of d0] {$e_1$};
      \node[state] (d2) [right=of d1] {$e_2$};
      \path[->]
        (d0) edge node {$\mathit{card?}$} (d1)
        (d1) edge node {$\mathit{pin?}$}  (d2)
        (d2) edge [bend left=15] node{$\mathit{money!}$} (d0);
      \path[->, delta]
        (d0) edge [loop above] node {$\delta_2$} (d0)
        (d1) edge [loop above] node {$\delta_2$} (d1);
    \end{tikzpicture}
    \caption[justification=raggedright, singlelinecheck=false]{The cash dispenser $\sys_{\textsc{disp}}$ as an LTS}
    \label{fig:disp-cash-lts}
  \end{subfigure}

  %\vspace{1.2em}
  \par
  % ----- (b) composed lts -----
\begin{subfigure}[b]{0.36\linewidth}
    \centering
    \resizebox{1\textwidth}{!}{
    \begin{tikzpicture}[
      >=Stealth, shorten >=1pt, auto, node distance=2.2cm and 1.2cm,
      on grid, semithick,
      every state/.style={draw, rounded corners, minimum width=8mm,
        minimum height=7mm, inner sep=1pt, font=\scriptsize, align=center,
        fill=blue!15, draw=blue!50!black},
      el/.style={font=\scriptsize, inner sep=1pt},
      delta/.style={dashed, draw=red!80, font=\scriptsize},
    ]
      \node[state, initial, initial text=] (q0) {$s_0,e_0$};
      \node[state] (q1) [below=of q0]      {$s_1,e_1$};
      \node[state] (q2) [below=of q1]      {$s_1,e_2$};
      \node[state] (q3) [below=of q2]      {$s_2,e_3$};
      \node[state] (q4) [left=of q3, xshift=-1cm] {$s_2,e_0$};
      \node[state] (q5) [right=of q3, xshift=1cm] {$s_0,e_3$};
      \path[->]
        (q0) edge node[el] {$\mathit{card?}$} (q1)
        (q1) edge node[el] {$\mathit{msg!}$}  (q2)
        (q2) edge node[el] {$\mathit{pin?}$}  (q3)
        (q3) edge node[el] {$\mathit{msg!}/\mathit{err!}$} (q4)
        (q3) edge node[el,swap] {$\mathit{money!}$} (q5)
        (q4) edge [bend left=20, looseness=1, in=-220] node[el,swap] {$\mathit{money!}$} (q0)
        (q5) edge [bend right=20, looseness=1, in=220] node[el] {$\mathit{msg!}/\mathit{err!}$} (q0);
      % \path[->, delta]
      %   (q0) edge [loop right] node {$\delta_1,\delta_2$} (q0)
      %   (q1) edge [loop right] node {$\delta_1$} (q1)
      %   (q2) edge [loop right] node {$\delta_1,\delta_2$} (q2)
      %   (q4) edge [out=60, in=30, looseness=10] node {$\delta_2$} (q4)  
      %   (q5) edge [out=150, in=120, looseness=10] node {$\delta_1$} (q5);
    \end{tikzpicture}}
    \caption{$\sys_{\textsc{UI}} \parop \sys_{\textsc{disp}}$ as an LTS}
    \label{fig:compose-lts}
  \end{subfigure}
  \hfill
  % ----- (c) lifted composition -----
  \begin{subfigure}[b]{0.62\linewidth}
    \centering
    
    %\resizebox{1\textwidth}{!}{
    \begin{tikzpicture}[
      >=Stealth, shorten >=1pt, auto, node distance=1.9cm and 1.6cm,
      on grid, semithick, 
      every state/.style={rectangle, align=center, rounded corners=5pt, minimum width=9mm,
        minimum height=9mm, inner sep=1pt, font=\scriptsize, align=center,
        fill=blue!15, draw=blue!50!black},
      el/.style={font=\scriptsize, align=center, inner sep=1pt},
      delta/.style={dashed, draw=red!80, font=\scriptsize, align=center},
    ]
      \node[state, initial, initial text=] (q0) {$s_0,e_0$\\$\scriptscriptstyle c_1\leq1$\\ $\scriptscriptstyle \wedge c_2\leq5$};
      \node[state] (q1) [below=of q0]      {$s_1,e_1$\\$\scriptscriptstyle c_1\leq1$\\ $\scriptscriptstyle \wedge c_2\leq5$};
      \node[state] (q2) [below=of q1]      {$s_1,e_2$\\$\scriptscriptstyle c_1\leq1$\\ $\scriptscriptstyle \wedge c_2\leq5$};
      \node[state] (q3) [below=of q2]      {$s_2,e_3$\\$\scriptscriptstyle c_1\leq1$\\ $\scriptscriptstyle \wedge c_2\leq5$};
      \node[state] (q4) [left=of q3, , xshift=-7mm] {$s_2,e_0$\\$\scriptscriptstyle c_1\leq1$\\ $\scriptscriptstyle \wedge c_2\leq5$};
      \node[state] (q5) [right=of q3, xshift=7mm] {$s_0,e_3$\\$\scriptscriptstyle c_1\leq1$\\ $\scriptscriptstyle \wedge c_2\leq5$};
      \path[->]
        (q0) edge node[el,swap] {$\scriptscriptstyle\mathit{card?}$\\$\scriptscriptstyle c_1<1\wedge c_2<5$\\$\scriptscriptstyle\{c_1,c_2\}$} (q1)
        (q1) edge node[el,swap] {$\scriptscriptstyle\mathit{msg!}$\\$\scriptscriptstyle c_1<1,\{c_1\}$} (q2)
        (q2) edge node[el,swap] {$\scriptscriptstyle\mathit{pin?}$\\$\scriptscriptstyle c_1<1\wedge c_2<5$\\$\scriptscriptstyle\{c_1,c_2\}$} (q3)
        (q3) edge node[el, swap] {$\scriptscriptstyle\mathit{msg!}/\mathit{err!}$} node[el] {$\scriptscriptstyle c_1<1$}  node[el,outer sep = 2mm] {$\scriptscriptstyle  \{c_1\}$} (q4)
        (q3) edge node[el] {$\scriptscriptstyle\mathit{money!}$} node[el,swap] {$\scriptscriptstyle c_2<5$} node[el,outer sep = -5mm] {$\scriptscriptstyle  \{c_2\}$} (q5)
        (q4) edge [bend left=18, looseness=1, in=-220] node[el] {$\scriptscriptstyle\mathit{money!}$\\$\scriptscriptstyle c_2<5,\{c_2\}$} (q0)
        (q5) edge [bend right=18, looseness=1, in=220] node[el, swap] {$\scriptscriptstyle\mathit{msg!}/\mathit{err!}$\\$\scriptscriptstyle c_1<1,\{c_1\}$} (q0);
      \path[->, delta]
        (q0) edge [loop right] node {$\scriptscriptstyle\delta_1, c_1 =1,\{c_1\}$\\$\scriptscriptstyle\delta_2,c_2=5,\{c_2\}$} (q0)
        (q1) edge [loop right] node[below, yshift=-0.2cm] {$\scriptscriptstyle\delta_2 , c_2=5,$\\$\scriptscriptstyle\{c_2\}$} (q1)
        (q2) edge [loop right] node[below, yshift=-0.3cm] {$\scriptscriptstyle\delta_1,c_1=1,\{c_1\}$\\$\scriptscriptstyle\delta_2, c_2=5,\{c_2\}$} (q2)
        (q4) edge [out=150, in=120, looseness=10]  node[above,xshift=0.3cm] {$\scriptscriptstyle\delta_1 , c_1=1,$\\$\scriptscriptstyle\{c_1\}$} (q4)
        (q5) edge [out=60, in=30, looseness=10] node[above,xshift=-0.3cm] {$\scriptscriptstyle\delta_2 , c_2=5,$\\$\scriptscriptstyle\{c_2\}$} (q5);
    \end{tikzpicture}
    %}
    \caption{TA $\chi^{[1]}(\sys_{\textsc{UI}})\!\parop\!\chi^{[5]}(\sys_{\textsc{disp}})$, lifted and composed with $\Mdec=[1,5]$}
    \label{fig:compose-ta}
  \end{subfigure}
  \caption{Parallel composition of UI and dispenser (the display $\sys_{\textsc{ui}}$ is shown in \Cref{fig:disp-lts}.
  (\subref{fig:disp-cash-lts}) the cash dispenser $\sys_{\textsc{disp}}$; (\subref{fig:compose-lts}) the LTS composition over states $(s_i,e_j)$; (\subref{fig:compose-ta}) its lifting with one clock per sub-component, i.e. $\Mdec=[1,5]$.}
  \label{fig:compose}
\end{figure}
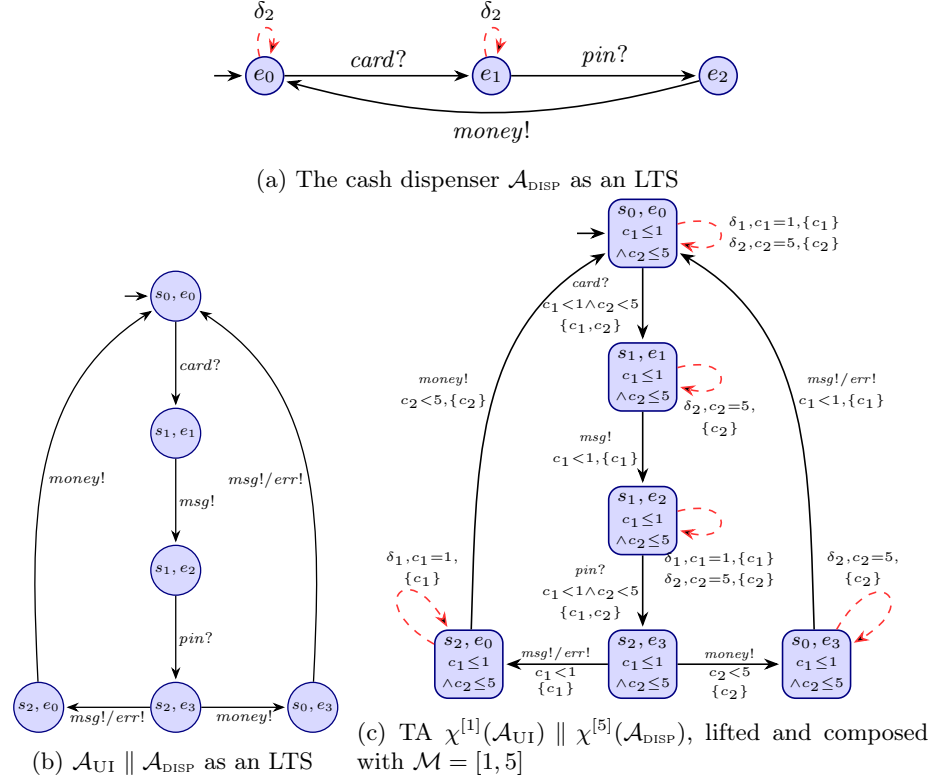

Parallel composition of multiple components is one of the main reasons for per-channel time-outs. A single global $M$ like in our previous work~\cite{BvdBGS25-not-blind} forces every component to ``share'' one quiescence deadline, so a fast component cannot conclude quiescence until the slowest one would. This increases the overall testing time. Per-channel clocks avoid this, because each channel concludes quiescence at its own bound. This keeps verdicts specific and avoids the time overhead of waiting out the slowest channel everywhere. We show that this structure survives composition, since the multi-channel lifting commutes with parallel composition.

\begin{restatable}[Compositionality]{theorem}{thmComp}
\label{thm:compose}
Let $\sys_1, \sys_2$ be LTSs with disjoint output alphabets and shared inputs, and let $\Mdec_1 \in \Rpos^{n_1}$,
$\Mdec_2 \in \Rpos^{n_2}$. Then
\[
  \chi^{\Mcat}(\sys_1 \parop \sys_2)
  \ =\
  \chi^{\Mdec_1}(\sys_1) \parop \chi^{\Mdec_2}(\sys_2).
\]
\end{restatable}

In other words, lifting a composed model gives exactly the composition of the lifted models (up to clock renaming). Thus, a modeller can specify each component in the untimed paradigm with its own bounds and compose the results without ever reasoning about a global time-out or building the multi-channel timed model by hand. With our definition of parallel composition and with the following lemma the conformance under composition then follows immediately.

\begin{restatable}[Compositionality of testable multi-ioco]{lemma}{lemMiocoCompose}
\label{lem:mioco-compose}
Let $\sys^1_I, \sys^1_S$ and let $\sys^2_I, \sys^2_S$ be the implementation--specification pairs of LTSs, with $\sys^1_I, \sys^2_I$ IOTSs. Assume they have disjoint outputs ($\ActOi{1} \cap \ActOi{2} = \emptyset$) and shared inputs, and have bound vectors $\Mdec_1 \in \RRplus^{n_1}$ and $\Mdec_2 \in \RRplus^{n_2}$.
If $\sys^1_I \: \tmiocoM{\Mdec_1} \: \sys^1_S$ and $\sys^2_I \: \tmiocoM{\Mdec_2} \: \sys^2_S$, then $\sys^1_I \parop \sys^2_I\: \tmiocoM{\Mcat} \: \sys^1_S \parop \sys^2_S$.
\end{restatable}

\begin{restatable}[Composition of timed conformance]{corollary}{corComp}
\label{cor:compose-conf}
Let $\sys^1_I, \sys^1_S$ and $\sys^2_I, \sys^2_S$ be implementation--specification pairs of LTSs, with $\sys^1_I, \sys^2_I$ IOTSs. Assume the pairs have disjoint outputs ($\ActOi{1} \cap \ActOi{2} = \emptyset$) and shared inputs, and have bound vectors $\Mdec_1 \in \RRplus^{n_1}, \Mdec_2 \in \RRplus^{n_2}$. If $\sys^1_I \: \tmiocoM{\Mdec_1} \: \sys^1_S$ and $\sys^2_I \: \tmiocoM{\Mdec_2} \: \sys^2_S$, then
$\chi^{\Mcat}(\sys^1_I \parop \sys^2_I)\ \mtioco{\Mcat}\
  \chi^{\Mcat}(\sys^1_S \parop \sys^2_S).$
\end{restatable}
\section{Related Work}
\label{sec:related}

Our contribution sits at the intersection of several extensions of \ioco: quiescence, timed conformance, multiple channels, and compositionality. Most prominently, we extend our own prior work~\cite{BvdBGS25-not-blind,BvdBGS26-not-blind} to a multi-channel setting. 

\begin{description}
    \item[Input-Output conformance and quiescence.] 
    Our work builds on the \ioco\ testing theory of Tretmans~\cite{T08}. Stokkink et al. later make quiescence a first-class citizen through quiescent transition systems~\cite{STS12} and later yet treat divergence explicitly~\cite{STS13}. In~\cite{TJ22}, Tretmans and Janssen revisit the foundations of the relation and note some shortcomings.
    \item[Multiple-channels.]
    Heerink~\cite{H98} introduces refusal testing with multiple input/output channels, and Brand\'an Briones and Brinksma~\cite{BBB05} give a multi-input/output relation; the first is untimed and operates on several input channels, the latter operates on timed-labelled transition systems. 
    \item[Timed testing.]
    Several timed-\ioco\ variants exist: \textbf{tioco} of Larsen et al.~\cite{LMN04}, \textbf{rtioco} of Krichen and Tripakis~\cite{KT04}, the bounded-quiescence $\tioco_M$ of Brand\'an Briones and Brinksma~\cite{BBB04}, and the liveness-preserving \textbf{ltioco} of Luthmann et al.~\cite{LGL19}, which, like us treats quiescence under composition but with a single global time-out and synchronisation hidden to internal actions.
    \item[Distributed systems.]
    Closest in spirit to our work is \emph{distributed} testing, where a system is observed through several interfaces. Hierons et al.\ adapt \ioco\ to this setting as \textbf{dioco}~\cite{HMN12}, and Gaston et al.~\cite{GHG13} give a timed distributed relation. The distinction is that distributed testing places independent (non-synchronising) testers at the ports, so the global order of events cannot be reconstructed. In contrast, our channels are observed by a single tester, so order and per-channel quiescence remain fully recoverable.
    \item[Compositionality.] 
    Compositionality of \ioco\ has been well-studied, but most of it is untimed. The seminal work of Van der Bijl et al.~\cite{BRT03} establishes that \ioco\ is compositional under certain restrictions; Daca et al.~\cite{DHKN14} study compositional specifications, and van Cuyck et al.~\cite{vCvAT23,vCvAT24} characterise compositionality via \emph{mutual acceptance}. Other compositions than parallel composition have been studied as well, e.g. merge and quotient by Bene\v{s} et al.~\cite{BDHKN15}, and sequential composition by Zameni et al.~\cite{ZvdBR25}.
\end{description}
\section{Conclusion}
\label{sec:conclusion}
We presented the multi-channel lifting $\chidec$ from untimed specifications to timed automata with one clock and one quiescence time-out \emph{per output channel} as an extension to our previous work~\cite{BvdBGS25-not-blind}. The lifting connects untimed multi-channel \mioco~\cite{H98} with the timed multi-channel relation $\mtioco{\Mdec}$~\cite{BBB05}  (\Cref{thm:preservation}). We show that it commutes with test generation and preserves verdicts (\Cref{thm:test-corr,thm:verdicts}). Our main contribution shows that it also commutes with shared-environment parallel composition (\Cref{thm:compose}), which means that conformance of independently specified components carries over to the composed timed system (\Cref{cor:compose-conf}). For the modeller, it means each component may be specified in the untimed world with its own per-channel bounds and  composed without the need of building the multi-channel timed model by hand. 

An immediate extension is to also admit internal $\tau$-actions and enable operations such as \emph{action hiding}. Internal steps let time elapse unobserved, so suspension traces and quiescence must be carefully reconstructed, as in an extension of our previous work~\cite{BvdBGS26-not-blind}.
A related direction is admitting output-to-input synchronisation in parallel composition as in~\cite{LT88}. This either requires input-enabledness on the synchronising actions or fragmentation of inputs into channels to keep quiescence componentwise. 

%- - - - - Acknowledgements - - - - - 
\paragraph{Acknowledgements.}
\ifanonymous
Acknowledgements omitted for double-blind review.
\else
This project has received funding from the European Union’s Horizon 2020 research and innovation programme under the Marie Skłodowska-Curie grant agreement No 101008233 (\textsc{MISSION}).
This research was supported by NWO project OCENW.M.23.155 \emph{Evidence-Driven Black-Box Checking (EVI)}.
\fi
\newpage
%- - - - - Bibliography - - - - - 
\bibliographystyle{abbrv}
\bibliography{thebib.bib}
%- - - - - - Appendix - - - - - - 
% \newpage
\appendix
% This is hacky, but after some tinkering the only way I got it to work..
\setcounter{section}{0}
\renewcommand{\thesection}{\Alph{section}}
\hypertarget{app:proofs-target}{}
\section{Appendix: Formal Proofs}
\label{app:proofs}
Below we provide detailed proofs to the theorems, corollaries and lemmas in the paper. The enumeration refers to the original one used in the paper.

\paragraph{Trace prefix.} Throughout, we use $\sqsubseteq$ to denote the trace prefix relation, i.e. given $\trace,\trace'\in\Act^*$ with $\trace=\action_1,\ldots,\action_k$ for $\trace'=\action_1,\ldots,\action_n$ for some $k\leq n$ we use the notation $\trace\sqsubseteq\trace'$ to denote that $\trace'$ is a \emph{subtrace} of $\trace$. 

%MGChange: new observation used by Thm 1, Thm 2, Thm 3 and Lemma 3

\paragraph{Uniformity of the lifting.}
The guard, reset and invariant that $\chidec$ adds on a transition depend only on the channel of its action and on $\Mdec$ (\Cref{def:lifting}), never on the system being lifted; and the $\delta_{\kchan}$ self-loop exists exactly at $\kchan$-quiescent states, which is what makes $\delta_{\kchan}$ a suspension action at that state in the first place. Consequently, for a suspension trace $\sigma$ that is a trace of two LTSs $\sys$ and $\sys'$, the timed realizations coincide:
$$\{\rho \in \SttracesM(\chidec(\sys)) \mid \rho{\downarrow}=\sigma\} = \{\rho \in \SttracesM(\chidec(\sys')) \mid \rho{\downarrow}=\sigma\}.$$
In particular, since the valuation reached after a timed trace $\rho$ is determined by $\rho$ (there are no $\tau$-steps), whether a step $(d,a)$ is admissible after $\rho$ depends only on $\rho$, $\Mdec$, and whether $a$ is enabled at the reached state in the untimed system. We refer to this as \emph{uniformity} of the lifting.

%--------------------------------------------------------
\lemCanonic*
%--------------------------------------------------------

\begin{proof}
%MGChange: the former LTS-to-TA direction has been removed; that direction now holds by \Cref{def:TStraces}.
The proof is by construction. Let
$\ttrace \in \SttracesM(\chidec(\sys))$. By~\Cref{def:mtioco-obs} any suspension timed trace can be written as
$$\location_0 \transition{(d_1, \action_1)} \location_1
  \transition{(d_2, \action_2)} \ldots
  \transition{(d_{m}, \action_{m})} \location_{m},$$
and, since $\ttrace \in \SttracesM(\chidec(\sys))$, every step is a genuine timed transition of $\chidec(\sys)$, i.e. its guards hold and its location invariants are respected throughout the elapsed delay. By~\Cref{def:lifting} every transition in $\trans_{\chidec(\sys)}$ extends a discrete transition of $\sys$ with a guard, an invariant, and a reset, or is an added $\delta_{\kchan}$ self-loop at a $\kchan$-quiescent location; explicitly, each has one of the forms:
\begin{itemize}
    \item $(\location, \iaction?, \bigwedge_{k} c_k < M_k, C,
          \location')$ for inputs
          $(\state, \iaction?, \state') \in \trans_{\sys}$ with
          $\iaction? \in \ActI$;
    \item $(\location, \oaction!, c_{\kchan} < M_{\kchan},
          \{c_{\kchan}\}, \location')$ for outputs
          $(\state, \oaction!, \state') \in \trans_{\sys}$ with
          $\oaction! \in \ActO^{\kchan}$;
    \item $(\location, \delta_{\kchan}, c_{\kchan} = M_{\kchan},
          \{c_{\kchan}\}, \location)$ for $\kchan$-quiescent locations.
\end{itemize}
The projection $\project{\ttrace}$ deletes the delays $d_j \in \RRnonzero$, which are the only information not present in $\sys$, leaving the sequence
$$\trace = \action_1 \action_2 \ldots \action_{m}.$$
We argue that $\trace \in \Straces(\sys)$ by following the suspended timed trace step-by-step. Since $L = S$, each $\location_j$ identifies a state $\state_j$ (w.l.o.g. otherwise rename the $\ell$ but keep the location-to-location transitions in place), and for each step we read off the matching transition of $\sys$, that is:
\begin{itemize}
  \item if $\action_j \in \ActI \cup \ActO$, the corresponding lifted
        transition extends a discrete transition
        $\state_{j-1} \transition{\action_j} \state_j$ of $\sys$, which
        is therefore present in $\trace$;
  \item if $\action_j = \delta_{\kchan}$, the step took the lifted
        self-loop, which by~\Cref{def:lifting} exists only at a
        $\kchan$-quiescent location. By \Cref{def:mtioco-obs} the step is admissible only if $\location_{j-1}$ is $\kchan$-quiescent, so $\state_{j-1}$ is $\kchan$-quiescent
        (cf.~\Cref{def:k-quiescent}) and contributes $\delta_{\kchan}$
        to its suspended traces. 
\end{itemize}
Hence every action of $\trace$ is admissible in $\sys$ in the order it
appears, so $\trace \in \Straces(\sys)$ with
$\project{\ttrace} = \trace$. Moreover, then by \Cref{def:TStraces} $\trace \in \TStraces(\sys)$. \qed
\end{proof}

%--------------------------------------------------------
\corIotsIota*
%--------------------------------------------------------
\begin{proof}
By \Cref{def:lifting}, the lifting adds on every input transition the guard $\bigwedge_{\kchan=1}^n c_{\kchan} < M_{\kchan}$. This guard is satisfied at exactly those valuations $v$ with $v_{\kchan} < M_{\kchan}$ for all $\kchan$, which is precisely the condition under which an IOTA is required to be input-enabled (\Cref{def:iota}). Since an IOTS enables every input at every state, and the lifting preserves states and adds these input transitions, every location of $\chidec(\sys_I)$ enables every input at every such valuation. Hence, if $\sys_I$ is an IOTS then $\chidec(\sys_I)$ is an IOTA. \qed
\end{proof}

%--------------------------------------------------------
\thmPreservation*
%--------------------------------------------------------

%MGChange: proof rewritten for \tmioco over testable extensions (\Cref{def:tmioco}).
\begin{proof}
Let $\Mdec \in \RRplus^{n}$ and write $\implTA = \chidec(\sys_I)$, $\specTA = \chidec(\sys_S)$. Since $\sys_I$ is an IOTS, $\implTA$ is an IOTA (\Cref{cor:iots-to-iota}). W.l.o.g.\ we identify each state with its location, since $L = S$ under $\chidec$ (\Cref{def:lifting}).

We first record how untimed and timed observables relate. Therefore, let $\rho \in \SttracesM(\specTA)$ and $\sigma = \rho{\downarrow}$; by \Cref{lem:canonic}, $\sigma \in \TStraces(\sys_S)$. Let $v$ be the valuation reached after $\rho$, which is the same in $\implTA$ and $\specTA$ (via uniformity of the lifting, see above). For $a \in \ActO^{\kchan} \cup \{\delta_{\kchan}\}$ and either system $\sys \in \{\sys_I, \sys_S\}$ with lifting $\systemTA$, \Cref{def:mtioco-obs} gives that some $(d,a) \in \outop^{\Mdec}_{\kchan}(\systemTA \afterop \rho)$ iff $\rho\cdot(d,a) \in \SttracesM(\systemTA)$ iff $\sigma \cdot a \in \TStraces(\sys)$ via $\rho\cdot(d,a)$. That is,
\begin{equation}
\label{eqn:obs-corr}
\{a \mid \exists d \ (d,a) \in \outop^{\Mdec}_{\kchan}(\systemTA \afterop \rho)\} \subseteq \outT_{\kchan}(\sys, \sigma),
\end{equation}
and conversely, if $\sigma\cdot a \in \TStraces(\sys)$ then by uniformity every realization of $\sigma$, in particular $\rho$, extends to a realization of $\sigma\cdot a$, so some $(d,a)$ lies in $\outop^{\Mdec}_{\kchan}(\systemTA \afterop \rho)$. Hence \eqref{eqn:obs-corr} is an equality. Finally, for a fixed $a$ the set of admissible $d$ after $\rho$ depends only on $v$, $\Mdec$ and the guard form of $a$ (uniformity), and is therefore the same for $\implTA$ and $\specTA$ whenever $a$ is enabled in both.

\medskip\noindent$\boxed{\Rightarrow}$
Assume $\sys_I \tmioco \sys_S$. Let $\rho \in \SttracesM(\specTA)$, $\sigma=\rho{\downarrow}$, and $(d,a) \in \outop^{\Mdec}_{\kchan}(\implTA \afterop \rho)$. By \eqref{eqn:obs-corr} for $\sys_I$, $a \in \outT_{\kchan}(\sys_I,\sigma)$; by $\tmioco$, $a \in \outT_{\kchan}(\sys_S,\sigma)$; by the converse of \eqref{eqn:obs-corr} for $\sys_S$, some $(d',a) \in \outop^{\Mdec}_{\kchan}(\specTA \afterop \rho)$, and since the admissible delays for $a$ coincide on both sides, $(d,a) \in \outop^{\Mdec}_{\kchan}(\specTA \afterop \rho)$. Hence $\implTA \mtioco{\Mdec} \specTA$.

\medskip\noindent$\boxed{\Leftarrow}$
Assume $\implTA \mtioco{\Mdec} \specTA$. Let $\sigma \in \TStraces(\sys_S)$ and $a \in \outT_{\kchan}(\sys_I,\sigma)$, i.e.\ $\sigma\cdot a \in \TStraces(\sys_I)$ via some $\rho\cdot(d,a) \in \SttracesM(\implTA)$ with $\rho{\downarrow} = \sigma$. As $\sigma \in \Straces(\sys_S)$, uniformity gives $\rho \in \SttracesM(\specTA)$. Now $(d,a) \in \outop^{\Mdec}_{\kchan}(\implTA \afterop \rho)$, so by $\mtioco{\Mdec}$ also $(d,a) \in \outop^{\Mdec}_{\kchan}(\specTA \afterop \rho)$, i.e.\ $\rho\cdot(d,a) \in \SttracesM(\specTA)$, with $\sigma\cdot a \in \TStraces(\sys_S)$ and $a \in \outT_{\kchan}(\sys_S,\sigma)$. Hence $\sys_I \tmioco \sys_S$.

\medskip
For the final claim, $\sys_I \mioco \sys_S$ implies $\sys_I \tmioco \sys_S$ (\Cref{def:tmioco}), hence $\implTA \mtioco{\Mdec} \specTA$ by the first direction.
\qed
\end{proof}

%--------------------------------------------------------
\thmTestCorr*
%--------------------------------------------------------

\begin{proof}
Recall that on test cases $\chidec$ decorates the existing transitions and adds none (\Cref{sec:tests}), so $\pass$ and $\fail$ remain terminal (specifically, no $\delta$ are added.
The proof is in two steps via set inclusion in both direction, i.e. case (1) $\chidec(\mathcal{T}_{\LTS}(\sys_S)) \subseteq \mathcal{T}_{\TA}(\chidec(\sys_S))$ and case (2) $\mathcal{T}_{\TA}(\chidec(\sys_S)) \subseteq \chidec(\mathcal{T}_{\LTS}(\sys_S))$. 

\medskip\noindent\emph{$\boxed{\chidec(\mathcal{T}_{\LTS}(\sys_S)) \subseteq \mathcal{T}_{\TA}(\chidec(\sys_S))}$}
Let $\test \in \mathcal{T}_{\LTS}(\sys_S)$ and put
$\testTA = \chidec(\test)$. We verify each clause of
\Cref{def:ta-test} to show that $\testTA\in\mathcal{T}_{\TA}(\chidec(\sys_S))$ is a timed test case.

\begin{description}
    \item[Structure.] $\chidec$ is structure preserving, so $\testTA$ has the same actions and the $\delta_{\kchan}$, already present in $\test$ by \Cref{def:lts-test}, the same tree shape, and the same finite, acyclic, deterministic structure. The $\delta_{\kchan}$ transitions in a test case lead to successor states, not self-loops (cf. \Cref{def:lts-test}), so acyclicity is preserved. Further, $\pass, \fail$ have no outgoing transitions in $\test$; $\chidec$ adds none, so they are terminal in $\testTA$. Every transition of $\testTA$ lies on a timed trace: a transition of $\test$ lies on some $\trace \in \traces(\test)$ ending in an input, $\pass$ or $\fail$; by input-specifiedness resp.\ soundness that trace (or its prefix before the final output) is in $\TStraces(\sys_S)$, hence realizable, and the final step is realizable by uniformity of the lifting.
    \item[Clocks and invariants.] By construction $\chidec$ adds the clocks $C = \{c_1,\dots,c_n\}$ and the invariant $\bigwedge_k c_k \le M_k$ at every non-terminal location, as required.
    \item[Observation/stimulation and per-channel guards.] 
    Each non-terminal state of $\test$ enables either all of $\ActO \cup \{\delta_1,\dots,\delta_n\}$ (observe) or $\ActO$ and one input (stimulate) (cf.~\Cref{def:lts-test}). Under $\chidec$ these become output transitions guarded by $c_{\kchan} < M_{\kchan}$ or $\delta_{\kchan}$ transitions guarded by $c_{\kchan} = M_{\kchan}$, or the output transitions plus a single input guarded $\bigwedge_k c_k < M_k$ with reset $C$. These are exactly the guard/reset forms required by \Cref{def:ta-test}.
    \item[Input-specifiedness, soundness, correctness.]
    For input-specifiedness, let $\ttrace \cdot (\tim, \iaction?) \in \ttraces(\testTA)$. Its projection $\ttrace{\downarrow} \cdot \iaction?$ lies in $\traces(\test)$, hence by input-specifiedness of $\test$ in $\TStraces(\sys_S) \subseteq \Straces(\sys_S)$. By uniformity, the realization $\ttrace \cdot (\tim, \iaction?)$ of this trace in $\testTA$ is also a realization in $\chidec(\sys_S)$, i.e.\ $\ttrace \cdot (\tim, \iaction?) \in \SttracesM(\chidec(\sys_S))$.
    Soundness transfers identically, i.e. a timed trace of $\testTA$ reaching $\pass$ projects to a trace of $\test$ reaching $\pass$, which lies in $\TStraces(\sys_S)$, and uniformity places the timed trace in $\SttracesM(\chidec(\sys_S))$.
    For correctness, let $\ttrace \cdot (\tim, \oaction) \in \ttraces(\testTA)$ reach $\fail$. Its projection reaches $\fail$ in $\test$, so by correctness of $\test$ it is not in $\TStraces(\sys_S)$; by \Cref{def:TStraces} it therefore has no realization in $\chidec(\sys_S)$, and therefore in particular $\ttrace \cdot (\tim, \oaction) \notin \SttracesM(\chidec(\sys_S))$.
\end{description}
All properties together yield $\chidec(\test)=\testTA\in\mathcal{T}_{\TA}(\chidec(\sys_S))$.

\medskip\noindent\emph{$\boxed{\chidec(\mathcal{T}_{\LTS}(\sys_S)) \supseteq \mathcal{T}_{\TA}(\chidec(\sys_S))}$}
Let $\testTA \in \mathcal{T}_{\TA}(\chidec(\sys_S))$ and let $\test=\project{\testTA}$, i.e. its \emph{projection} by erasing clocks, invariants, guards, and resets, but keeping locations (resp. states), the initial location, and the labelled transition relation. We now show $\chidec(\test) = \testTA$ and $\test \in \mathcal{T}_{\LTS}(\sys_S)$.

By \Cref{def:ta-test}, every transition of $\testTA$ has one of the
three canonic forms (1) input guarded by $\bigwedge_k c_k < M_k$ reset $C$; (2) channel-$\kchan$ output guarded $c_{\kchan} < M_{\kchan}$ reset $\{c_{\kchan}\}$, or (3) $\delta_{\kchan}$ guarded $c_{\kchan} = M_{\kchan}$ reset $\{c_{\kchan}\}$, and every non-terminal location carries $\bigwedge_k c_k \le M_k$. 
These are precisely the decorations $\chidec$ adds (\Cref{def:lifting}). Since the projection $\test$ keeps exactly the discrete transitions and $\delta_{\kchan}$ occurs in $\testTA$ as an ordinary transition (rather than being added by $\chidec$), re-applying $\chidec$ restores every guard, reset and invariant. Hence $\chidec(\test) = \testTA$.

The projection $\project{\testTA}$ is tree-shaped, finite, acyclic, and deterministic because $\testTA$ is, these being properties of the discrete structure. The terminal states and the observe/stimulate options are read off directly: an observe location (enabling all outputs and all $\delta_{\kchan}$) projects to an observe state; a stimulate location (enabling all outputs and one input) projects to a stimulate state. For the remaining conditions, note that since every transition of $\testTA$ lies on a timed trace (\Cref{def:ta-test}), every $\trace \in \traces(\test)$ is the projection of some $\ttrace \in \ttraces(\testTA)$. Input-specifiedness: if $\trace\cdot\iaction? \in \traces(\test)$ then some $\ttrace\cdot(\tim,\iaction?) \in \ttraces(\testTA)$ projects to it, which by input-specifiedness of $\testTA$ lies in $\SttracesM(\chidec(\sys_S))$, so $\trace\cdot\iaction? \in \TStraces(\sys_S)$ by \Cref{def:TStraces}. Soundness is identical. Correctness: if $\trace\cdot\oaction \in \traces(\test)$ reaches $\fail$, then every realization $\ttrace\cdot(\tim,\oaction) \in \ttraces(\testTA)$ reaches $\fail$ and by correctness of $\testTA$ is not in $\SttracesM(\chidec(\sys_S))$; as this holds for every realization, $\trace\cdot\oaction \notin \TStraces(\sys_S)$. Hence $\test \in \mathcal{T}_{\LTS}(\sys_S)$, and $\testTA = \chidec(\test) \in \chidec(\mathcal{T}_{\LTS}(\sys_S))$.

\medskip
Both inclusions hold, so
$\chidec(\mathcal{T}_{\LTS}(\sys_S)) =
\mathcal{T}_{\TA}(\chidec(\sys_S))$. \qed
\end{proof}

%--------------------------------------------------------
\thmVerdicts*
%--------------------------------------------------------

\begin{proof}
Let $\Mdec\in\RRplus^n$ and write $\implTA = \chidec(\sys_I)$ and $\specTA = \chidec(\sys_S)$. Recall from LTS-test verdicts (text below \Cref{def:lts-test}) that $\sys_I$ fails a test $\test$ iff some $\trace \in \TStraces(\sys_I) \cap \traces(\test)$ reaches $\fail$, and passes otherwise; likewise $\implTA$ fails a timed test $\testTA$ iff some $\ttrace \in \SttracesM(\implTA) \cap \ttraces(\testTA)$ reaches $\fail$. Since a test case carries $\delta_1,\dots,\delta_n$ in its alphabet (\Cref{def:lts-test}), its suspension traces coincide with its traces.

By \Cref{thm:test-corr}, every timed test case for $\chidec(\sys_S)$ is $\chidec(\test)$ for a $\test \in \mathcal{T}_{\LTS}(\sys_S)$, and vice versa. We use this bijection to prove the preservation of verdicts, i.e. for a test $\test$ and its lifting $\testTA = \chidec(\test)$, a testable trace reaches $\fail$ in $\test$ while being a trace of $\sys_I$ iff its timed trace reaches $\fail$ in $\testTA$ while being a trace of $\implTA$:
\begin{equation}
\begin{split}
  & \exists\, \trace \in \TStraces(\sys_I)\cap\traces(\test):
    \test \transition{\trace} \fail \\
  \iff\ &
  \exists\, \ttrace \in \SttracesM(\implTA)\cap\ttraces(\testTA):
    \testTA \Transition{\ttrace} \fail .
\end{split}
\label{eqn:fail-corr}
\end{equation}
\noindent$\boxed{\Rightarrow}$ Given is
$\trace\in\TStraces(\sys_I)\cap\traces(\test)$ with $\test\xrightarrow{\trace}\fail$.
By \Cref{def:TStraces} there is $\ttrace \in \SttracesM(\implTA)$ with $\ttrace{\downarrow} = \trace$. Since $\trace \in \traces(\test)$ and $\testTA = \chidec(\test)$ adds the same canonic guards as $\implTA$, uniformity gives $\ttrace \in \ttraces(\testTA)$. Since $\fail$ is reached in $\test$ along $\trace$, it is reached in $\testTA = \chidec(\test)$ along $\ttrace$, because $\chidec$ maps the discrete transition into $\fail$ to the corresponding timed transition into $\fail$.

\medskip\noindent$\boxed{\Leftarrow}$ This case is symmetrical and we project $\ttrace$ by ${\downarrow}$. The resulting projection $\trace$ is in $\TStraces(\sys_I)$ (by \Cref{lem:canonic}) and a trace of $\test$ (since $\testTA = \chidec(\test)$ retains exactly the discrete transitions), and reaches $\fail$ in $\test$.

\medskip\noindent\emph{(1) Preservation of passing.}
The proof is by contraposition. Suppose $\implTA = \chidec(\sys_I)$ fails $\mathcal{T}_{\TA}(\chidec(\sys_S))$. Then some $\testTA \in \mathcal{T}_{\TA}(\chidec(\sys_S))$ has a trace of $\implTA$ reaching $\fail$, i.e.\ the right side of \eqref{eqn:fail-corr} holds. By \Cref{thm:test-corr}, $\testTA = \chidec(\test)$ for some $\test \in \mathcal{T}_{\LTS}(\sys_S)$, so \eqref{eqn:fail-corr} gives a $\sys_I$-trace reaching $\fail$ in $\test$. Hence $\sys_I$ fails $\test$, and therefore does not pass $\mathcal{T}_{\LTS}(\sys_S)$.

\medskip\noindent\emph{(2) Preservation of failing.}
Suppose $\sys_I$ fails $\mathcal{T}_{\LTS}(\sys_S)$, i.e.\ some $\test \in \mathcal{T}_{\LTS}(\sys_S)$ has a trace of $\sys_I$ reaching $\fail$ in $\test$. Then $\testTA = \chidec(\test) \in \mathcal{T}_{\TA}(\chidec(\sys_S))$ by \Cref{thm:test-corr}, and the left side of \eqref{eqn:fail-corr} holds, so the right side gives an $\implTA$-trace reaching $\fail$ in $\testTA$. Hence $\implTA$ fails $\testTA$, and therefore fails $\mathcal{T}_{\TA}(\chidec(\sys_S))$. \qed
\end{proof}

%--------------------------------------------------------
\thmComp*
%--------------------------------------------------------

\begin{proof}
Let $\sys_1 \parop \sys_2 = \langle S_1 \times S_2, \Act,
\to_{\parop}, (s_{0,1}, s_{0,2})\rangle$ be the composed system as per \Cref{def:par-lts} and recall that its output partition is the disjoint union of the two component partitions, with $\sys_1$ channels indexed $1,\ldots,n_1$ and $\sys_2$'s channels indexed $n_1+1, \dots, n_1+n_2$.
Under this indexing the bound vector of the composition is exactly $\Mcat$, and $\chidec[\Mcat]$ introduces clocks $c_1, \dots, c_{n_1+n_2}$. We show the two sides of the equation have identical locations, clocks, transitions and invariants, which ultimately proves their equivalence.

\begin{description}
    \item[Locations and clocks.] Both sides have location set $S_1 \times S_2$ and initial location $(s_{0,1}, s_{0,2})$: the left by $L = S$ of $\chidec$ (\Cref{def:lifting}) applied to $\sys_1 \parop \sys_2$, and the right by the product of \Cref{def:par-ta}. Both have clock set $\{c_1,\dots,c_{n_1+n_2}\}$: the left because the composed partition has $n_1+n_2$ channels, the right because \Cref{def:par-ta} takes the union of the component clock sets. The clock indexing coincides by the convention above.
    \item[Input transitions.] Recall that we stipulate that all inputs are shared between all composing systems, and let $i? \in \ActI^{\mathrm{sync}}$ be such a shared input. On the left, $\sys_1 \parop \sys_2$ synchronises it (\Cref{def:par-lts}): $(\ell_1,\ell_2) \transition{i?} (\ell_1',\ell_2')$ iff $\ell_1 \transition{i?} \ell_1'$ and $\ell_2 \transition{i?} \ell_2'$. Lifting, $\chi^{\Mcat}$ gives guard $\bigwedge_{k=1}^{n_1+n_2} c_k < (\Mcat)_k$ and reset of all clocks. On the right, $i?$ is lifted within each component, i.e. guard $\bigwedge_{k=1}^{n_1} c_k < (\Mdec_1)_k$ and $\bigwedge_{k=1}^{n_2} c_{n_1+k} < (\Mdec_2)_k$, each resetting all of its own clocks. The synchronisation clause of \Cref{def:par-ta} joins the guards and unifies the resets. The joined guard is $\bigwedge_{k=1}^{n_1+n_2} c_k < (\Mcat)_k$ and the unified reset is all of $\{c_1,\dots,c_{n_1+n_2}\}$, matching the left. 
    \item[Output transitions.] Let $o! \in \ActO^{\kchan}(\sys_1)$, so $\kchan \le n_1$. Outputs are disjoint, so $o!$ belongs to $\sys_1$ and is not shared with $\sys_2$ by assumption. On the left, $o!$ is a transition of $\sys_1 \parop \sys_2$ via the interleaving clause of \Cref{def:par-lts}: $(\ell_1,\ell_2) \transition{o!} (\ell_1',\ell_2)$ iff $\ell_1 \transition{o!} \ell_1'$ in $\sys_1$. Lifting it, $\chi^{\Mcat}$ gives guard $c_{\kchan} < (\Mcat)_{\kchan}$$ = (\Mdec_1)_{\kchan}$ and reset $\{c_{\kchan}\}$. On the right, the same $\ell_1 \transition{o!} \ell_1'$ is first lifted by $\chi^{\Mdec_1}$ to guard $c_{\kchan} < (\Mdec_1)_{\kchan}$, reset $\{c_{\kchan}\}$, then carried into the product by the interleaving clause of \Cref{def:par-ta}. The two transitions coincide. Outputs of $\sys_2$ are symmetric, with channel index shifted by $n_1$.
    \item[Quiescence self-loops.] Let $\kchan \le n_1$ (the case $\kchan > n_1$ is symmetric). We claim $(\ell_1,\ell_2)$ is $\kchan$-quiescent in $\sys_1 \parop \sys_2$ iff $\ell_1$ is $\kchan$-quiescent in $\sys_1$. A channel-$\kchan$ output can be enabled at $(\ell_1,\ell_2)$ only via the interleaving clause from $\sys_1$ (outputs are disjoint, so $\sys_2$ has no channel-$\kchan$ output; shared inputs synchronise to inputs and not outputs, and by \Cref{def:par-lts} no output of one component is an input of the other). Hence a channel-$\kchan$ output is enabled at $(\ell_1,\ell_2)$ iff it is enabled at $\ell_1$, so the two states agree on $\kchan$-quiescence. On the left, $\chi^{\Mcat}$ therefore adds the self-loop $\langle (\ell_1,\ell_2), \delta_{\kchan}, c_{\kchan} = (\Mcat)_{\kchan}, \{c_{\kchan}\}, (\ell_1,\ell_2)\rangle$ iff $\ell_1$ is $\kchan$-quiescent. On the right, $\chidec[\Mdec_1]$ adds $\langle \ell_1, \delta_{\kchan}, c_{\kchan} = (\Mdec_1)_{\kchan}, \{c_{\kchan}\}, \ell_1\rangle$ iff $\ell_1$ is $\kchan$-quiescent, and \Cref{def:par-ta} carries this self-loop into the product as a $\delta_{\kchan}$ loop at $(\ell_1,\ell_2)$. Since it is non-shared  it interleaves and leaves $\ell_2$ fixed. Since $(\Mcat)_{\kchan} = (\Mdec_1)_{\kchan}$, the two quiescent self-loops coincide.
    \item[Invariants.] On the left, the lifting $\chi^{\Mcat}$ assigns to every location the invariant $\bigwedge_{k=1}^{n_1+n_2} c_k \le (\Mcat)_k$ where $(\Mcat)_k$ is the $k$-th position of the vector. On the right, \Cref{def:par-ta} assigns $(\ell_1,\ell_2)$ the conjunction of the component invariants, $\bigwedge_{k=1}^{n_1} c_k \le (\Mdec_1)_k \wedge \bigwedge_{k=1}^{n_2} c_{n_1+k} \le (\Mdec_2)_k$. Since $(\Mcat)_k = (\Mdec_1)_k$ for $k \le n_1$ and $(\Mcat)_{n_1+k} = (\Mdec_2)_k$ for $k \le n_2$, the two invariants are identical.
\end{description}

The two timed automata are identical because all four ingredients coincide.
\qed
\end{proof}

%--------------------------------------------------------
\lemMiocoCompose*
%--------------------------------------------------------

%MGChange: proof rewritten for \tmiocoM with explicit bound vectors.
\begin{proof}
Let $I = \sys^1_I \parop \sys^2_I$ and $S = \sys^1_S \parop \sys^2_S$. By \Cref{def:tmioco} we must show, for every $\sigma \in \TStraces[\Mcat](S)$ and every channel $\kchan \in \{1,\dots,n_1+n_2\}$, that $\outT_{\kchan}[\Mcat](I,\sigma) \subseteq \outT_{\kchan}[\Mcat](S,\sigma)$.

\emph{Projection.} For a suspension trace $\sigma$ of $S$ and $i \in \{1,2\}$, let $\sigma|_i$ be its projection onto $\Actd_{\sys^i}$. By \Cref{def:par-lts}, shared inputs advance both components and every other action advances exactly one, so $\sigma|_i \in \Straces(\sys^i_S)$, and $S \afterop \sigma$ consists of pairs $(s_1,s_2)$ with $s_i \in \sys^i_S \afterop \sigma|_i$; likewise for $I$. Moreover $\sigma|_i$ is testable \emph{for $\Mdec_i$}: a witness $\rho \in \SttracesMcus{\Mcat}(\chi^{\Mcat}(S))$ for $\sigma$ is, by \Cref{thm:compose}, a timed trace of $\chi^{\Mdec_1}(\sys^1_S) \parop \chi^{\Mdec_2}(\sys^2_S)$, and its projection $\rho|_i$ onto component $i$ (dropping the other component's actions and accumulating their delays) is a timed trace of $\chi^{\Mdec_i}(\sys^i_S)$ with $\rho|_i{\downarrow} = \sigma|_i$. Hence $\sigma|_i \in \TStraces[\Mdec_i](\sys^i_S)$, and the same argument shows that $\sigma\cdot a \in \TStraces[\Mcat](I)$ implies $\sigma|_i\cdot a \in \TStraces[\Mdec_i](\sys^i_I)$ for any $a$ owned by component $i$. Note that the bound of a channel is unchanged by concatenation, i.e. if channel $\kchan$ of the composition is owned by component $i$, then $(\Mcat)_{\kchan} = (\Mdec_i)_{\kchan'}$ for the corresponding index $\kchan'$ in that component. 

We note that the projection is well defined: component $i$'s clocks are reset only by its own outputs, by $\delta_{\kchan}$ with $\kchan$ owned by $i$, and by inputs, which are shared and hence kept. Dropped steps therefore only let time pass, so component $i$'s valuation evolves identically in $\rho$ and $\rho|_i$. Moreover the composite guards and invariant imply the component ones (they are conjunctions over a superset of channels), so every kept step remains admissible at its accumulated delay, and the component invariant holds throughout it.

\emph{Factoring.} Let channel $\kchan$ be owned by component $i$; ownership is unique since output partitions are disjoint. As in the proof of \Cref{thm:compose}, a channel-$\kchan$ output is enabled at $(s_1,s_2)$ iff it is enabled at $s_i$, and $(s_1,s_2)$ is $\kchan$-quiescent iff $s_i$ is.

\emph{Conclusion.} Fix $\sigma \in \TStraces[\Mcat](S)$, a channel $\kchan$ owned by $i$, and $a \in \outT_{\kchan}[\Mcat](I,\sigma)$, i.e.\ $\sigma\cdot a \in \TStraces[\Mcat](I)$ via some $\rho\cdot(d,a) \in \SttracesM[\Mcat](\chi^{\Mcat}(I))$. By projection, $\sigma|_i \in \TStraces[\Mdec_i](\sys^i_S)$ and $\sigma|_i \cdot a \in \TStraces[\Mdec_i](\sys^i_I)$, so $a \in \outT_{\kchan}[\Mdec_i](\sys^i_I,\sigma|_i)$. By the hypothesis $\sys^i_I \tmiocoM{\Mdec_i} \sys^i_S$, we get $a \in \outT_{\kchan}[\Mdec_i](\sys^i_S,\sigma|_i)$, in particular $\sigma|_i\cdot a \in \Straces(\sys^i_S)$, so $a$ is enabled (resp.\ $\sys^i_S$ is $\kchan$-quiescent) at $\sys^i_S \afterop \sigma|_i$, and by factoring the same holds at $S \afterop \sigma$. It remains to realize $\sigma\cdot a$ in $\chi^{\Mcat}(S)$. By uniformity, $\rho \in \SttracesM[\Mcat](\chi^{\Mcat}(S))$ (as $\sigma \in \Straces(S)$), and the step $(d,a)$ is admissible after $\rho$ in $\chi^{\Mcat}(S)$ iff $a$ is enabled at $S \afterop \sigma$ and the guard and invariant, which are the same canonic constraints as in $\chi^{\Mcat}(I)$, hold for $d$; both conditions are met. Hence $\rho\cdot(d,a) \in \SttracesM[\Mcat](\chi^{\Mcat}(S))$, so $\sigma\cdot a \in \TStraces[\Mcat](S)$ and $a \in \outT_{\kchan}[\Mcat](S,\sigma)$. As $\sigma$, $\kchan$ and $a$ were arbitrary, $I \tmiocoM{\Mcat} S$.
\qed
\end{proof}

%--------------------------------------------------------
\corComp*
%--------------------------------------------------------
\begin{proof}
By \Cref{lem:mioco-compose}, $\sys^1_I \parop \sys^2_I\: \tmiocoM{\Mcat}\: \sys^1_S \parop \sys^2_S$. Further, $\sys^1_I \parop \sys^2_I$ is an IOTS because parallel composition preserves input-enabledness (all inputs are shared and enabled in every state of both components, so every synchronised input is enabled in every product state). Then, by \Cref{cor:iots-to-iota} its lifting is an IOTA and $\mtioco{\Mcat}$ is well-typed. Applying \Cref{thm:preservation} to the composed systems with bound vector $\Mcat$ gives
$$\chi^{\Mcat}(\sys^1_I \parop \sys^2_I)\: \mtioco{\Mcat}\:
 \chi^{\Mcat}(\sys^1_S \parop \sys^2_S).$$ \qed
\end{proof}
\end{document}